\documentclass{article}

\usepackage{microtype}
\usepackage{graphicx}
\usepackage{subcaption}
\usepackage{booktabs} 
\usepackage{enumitem}

\usepackage{hyperref}

 \usepackage[accepted]{icml2026}

\usepackage{amsmath}
\usepackage{amssymb}
\usepackage{mathtools}
\usepackage{amsthm}
\usepackage{thmtools}
\usepackage{empheq}

\newcommand{\predfw} {\widehat{w}}
\newcommand{\cost}{\textnormal{cost}}
\newcommand{\OPT}{\textsc{opt}}
\newcommand{\ALG}{\textsc{alg}}
\newcommand{\algint}{\ALG_{\mathrm{int}}}
\newcommand{\algfrac}{\ALG_{\mathrm{frac}}}

\newcommand{\reals}{\mathbb{R}}
\newcommand{\E}{\mathbb{E}} 

\DeclareMathOperator*{\argmin}{\arg\!\min}

\usepackage[capitalize,noabbrev]{cleveref}

\theoremstyle{plain}
\newtheorem{theorem}{Theorem}[section]

\newtheorem{lemma}[theorem]{Lemma}

\theoremstyle{definition}
\newtheorem{definition}[theorem]{Definition}

\theoremstyle{remark}

\usepackage[textsize=tiny]{todonotes}

\begin{document}

\twocolumn[
  \icmltitle{
  Learning-Augmented Online Minimization with Dual Predictions
  }



  \icmlsetsymbol{equal}{*}

  \begin{icmlauthorlist}
    \icmlauthor{Christian Coester}{equal,oxf}
    \icmlauthor{Alexa Tudose}{equal,oxf}
    \icmlauthor{Alexander Turoczy}{equal,oxf}
  \end{icmlauthorlist}

  \icmlaffiliation{oxf}{University of Oxford, United Kingdom}
  \icmlcorrespondingauthor{Christian Coester}{christian.coester@cs.ox.ac.uk}
  \icmlcorrespondingauthor{Alexa Tudose}{maria-alexa.tudose@cs.ox.ac.uk}
  \icmlcorrespondingauthor{Alexander Turoczy}{alexanderturoczy@gmail.com}

  \icmlkeywords{Online algorithms, learning-augmented algorithms, algorithms with predictions, metrical task systems, laminar set cover, parking permit problem, k-server, learned duals, dual predictions}

  \vskip 0.3in
]



\printAffiliationsAndNotice{\icmlEqualContribution}

\begin{abstract}
We present learning-augmented algorithms for two general classes of online minimization problems: metrical task systems and laminar set cover. Both algorithms achieve improved theoretical guarantees using machine-learned predictions of an optimal solution to the dual linear program. Unlike optimal primal solutions, which can change drastically under tiny instance perturbations, these dual solutions are much more stable, which ensures the existence of good (and learnable) predictions for families of similar instances. While previous work has used dual predictions in offline settings and for online maximization problems, our algorithms are, to the best of our knowledge, the first demonstration that such dual predictions can be effective for online minimization. Our theoretical results are complemented by experiments on the $k$-server problem and the parking permit problem.

\end{abstract}

\section{Introduction}

Many decision-making problems arise in settings where inputs arrive over time (e.g., as a sequence of requests) and decisions must be made without knowing the future. Such problems are commonly modeled as \emph{online} optimization problems, where an algorithm must act irrevocably as each request arrives. The classical analysis of online algorithms assumes that the input is completely unknown and may even be chosen adversarially. Online algorithms are then evaluated by their \emph{competitive ratio}, defined as the worst-case ratio between the cost of the online algorithm and the best solution in hindsight. This viewpoint has led to a rich theory, and to strong impossibility results, but it can feel misaligned with practice, where similar instances of the \emph{same} optimization task are solved day after day. Modern systems often have access to historical data, forecasts, or learned models that can try and predict aspects of the future.

Motivated by this, the field of \emph{learning-augmented algorithms} (a.k.a.\ \emph{algorithms with predictions}) has emerged in recent years~\cite{LykourisV21,PurohitSK18,MitzenmacherV20}. Here, an algorithm's input is augmented with predictions about the instance; the goal is to design algorithms that benefit from predictions if these are reasonably accurate, while still satisfying classical worst-case bounds even when predictions are highly erroneous.

This paradigm has spurred progress across many online problems, including paging/caching, rent-or-buy problems, scheduling, and graph and metric optimization. Learning-augmented algorithms have also been developed in many other areas of research, such as to improve running times, data structures, approximation algorithms, streaming algorithms, and in game theory. We refer to \citet{LindermayrM26} for an extensive repository of research on learning-augmented algorithms. 


\medskip
\paragraph{What to Predict?}
A central design choice is to determine what information should be predicted. Common approaches for online minimization problems have been \emph{event predictions} (e.g., time or location of future requests) and \emph{action predictions} (i.e., suggested actions that the algorithm should take, intended to represent a near-optimal primal solution to the optimization problem). For example in caching, event predictions in the form of next request times for each page have been used successfully~\cite{LykourisV21,Rohatgi20,Wei20}.

However, these types of predictions come with some drawbacks. One drawback of event predictions for caching is that they are \emph{unstable}: A single change in an otherwise perfectly predicted sequence can lead to an unbounded prediction error (because an anticipated event may be delayed indefinitely). A consequence of this instability is that a \emph{single} incorrect prediction can completely deteriorate the performance of the state-of-the-art learning-augmented algorithm for the problem, making it perform as poorly as algorithms \emph{without} predictions. We show an example of such behavior in Appendix \ref{sec:unstablecaching}. Notably, this does not contradict \emph{smoothness} (graceful degradation with prediction error): smoothness bounds performance in terms of the error, but does not ensure the error itself is stable under small instance changes.

Beyond their instability, another downside of these event predictions for caching is that already for the slightly more general problem of weighted caching (and consequently further generalizations like $k$-server and metrical task systems), even \emph{exact knowledge} of the next request time of each page/location was shown to be insufficient for improved algorithmic guarantees~\cite{BansalCKPV22,JiangPS22}.

Action predictions suffer from a similar problem of instability: optimal primal solutions can be extremely fragile. Small perturbations to an instance can completely change the structure of an optimal primal assignment (see Appendix \ref{sec:instabprimalpred}). As a consequence, it is often hard to justify why action predictions would be available (i.e., why they are \emph{learnable}). While the usefulness of action predictions is intuitively evident (if you get good advice of what to do, just do it), action predictions shift an enormous burden onto the predictor, as it is a strong assumption that high-quality action advice is given.

We advocate a different prediction interface: \emph{learn the dual}. Concretely, we study online minimization problems that admit natural linear programming (LP) formulations and design learning-augmented algorithms whose predictions are assignments to variables of the corresponding dual LP. Our thesis is that for a wide range of problems, with an appropriate LP formulation and prediction-error notion, dual predictions can simultaneously satisfy three desiderata:
\begin{enumerate}
  \item \textbf{Stability:} small changes in the instance cause only small prediction error;
  \item \textbf{Usefulness:} small prediction error implies near-optimal performance, and the guarantee degrades smoothly with the error (and can be bounded robustly even for very large error);
  \item \textbf{Learnability:} the prediction target can be learned from a polynomial number of samples from a distribution.
\end{enumerate}
At a high level, stability and learnability play complementary roles: while learnability ensures that the best prediction for a given distribution can be identified with small sample complexity, stability justifies that a good prediction exists in the first place for a family of similar instances. 

We instantiate this approach in two fairly general online minimization domains.

\paragraph{Laminar Set Cover.}
Set cover is a central problem in online algorithms, and many competitive algorithms for seemingly different problems are derived via equivalent reformulations as (set) covering problems~\cite{BuchbinderN09}. 
We study a special case where the set family is laminar, a constraint that is natural in ``interval-like'' covering settings. This captures problems such as ski rental~\cite{PhillipsW99}, dynamic power management~\cite{dynamicpowermanagement}, the parking permit problem (PPP)~\cite{Meyerson05, coester2026a}, the multi parking permit problem \cite{multippp}, the Steiner leasing problem \cite{Meyerson05}, and sum-radii clustering in ultrametrics (a.k.a. HSTs) \cite{sumradiiclustering}. For general laminar set cover, the optimal competitive ratio is 
$\Theta(\log f)$, where $f$ is the maximum number of sets containing an element: the lower bound follows from the parking permit problem \cite{Meyerson05}, the upper bound from fractional set cover \cite{AlonAABN09} combined with online rounding (Lemma~\ref{lemma:rounding-randomized-online}). 

Another closely related problem is the Weighted Path Augmentation Problem (WPAP) \cite{WPAP1}, also known as Interval Set Cover. In this problem, the universe consists of points on $\mathbb R$, and the set family consists of intervals of $\mathbb R$. In Lemma \ref{lem:laminarization} of the Appendix, we show that WPAP is equivalent to laminar set cover up to a constant factor in the competitive ratio. Our rounding algorithm (Lemma~\ref{lemma:rounding-randomized-online}) thus confirms a conjecture by \citet{WPAP1} about the existence of such a rounding algorithm for WPAP.

\paragraph{Metrical Task Systems (MTS).}
MTS is a foundational online problem introduced by Borodin, Linial, and Saks \cite{BorodinLS92} and extensively studied since. It is a broad generalization of many online problems, including caching, $k$-server, convex function chasing, dynamic power management, and layered graph traversal, and with links to the experts problem in online learning \cite{BlumB00}. In an MTS instance, an algorithm moves in a metric space to service requests, and the goal is to minimize service and movement costs. In the classical setting without predictions, for any $n$-point metric, the competitive ratio of deterministic algorithms is exactly $2n-1$~\cite{BorodinLS92} and for randomized algorithms it is between $\Omega(\log n)$ and $O(\log^2 n)$, where both the lower and upper bounds are tight for some metrics~\cite{BubeckCLL21,CoesterL22,BubeckCR23}. Better bounds are achievable for special cases of MTS, as these correspond to restricting the class of cost functions.


\subsection{Our Results}

\paragraph{Notation.} We specify some standard notation we will use throughout the paper. We use $\textsc{alg}$ to denote either a given algorithm or its cost on a given instance, and $\textsc{opt}$ for the optimal offline cost. \textsc{alg} is $\mathcal R$-competitive if $\mathbb E[\textsc{alg}]\le \mathcal R\cdot\textsc{opt}$ on any instance. We use $\eta$ to denote a measure of error for a given prediction, which is defined in a problem-specific manner. We also define the normalized $\bar \eta := \frac{\eta}{\textsc{opt}}$.

We prove that dual predictions for both Laminar Set Cover and Metrical Task Systems fulfill the desired qualities of usefulness, stability, and learnability. Our results are summarized in the theorems below.

\begin{restatable}{theorem}{thmlaminarsetcover}{(Informal).} \label{thm:laminarsetcover}
    Suppose there is an $\mathcal R$-competitive online algorithm for a class of laminar set cover instances. Then predictions $\hat{\mathbf y}$ to an optimal dual solution satisfy
    \begin{enumerate}[label=(\alph*)]
        \item (Usefulness) For any constant $\epsilon > 0$, there exists a learning-augmented algorithm which satisfies \\ $\mathbb{E}[\textsc{alg}] = (1+\epsilon)\textsc{opt}+O\left(\frac{\mathcal R \eta}{\epsilon}\right)$.
        \item \label{item:stability}(Stability) For problem instances $X_1$ and $X_2$, there exist optimal duals $\mathbf y^*_1, \mathbf y^*_2$ satisfying $\lVert \mathbf y^*_1 - \mathbf y^*_2 \rVert_1 = O(|X_1 \Delta X_2|)$, where $\Delta $ denotes symmetric difference.
        \item (Learnability) The predictions are learnable with polynomial sample complexity.
        \item (Culmination) With access to a polynomial number of i.i.d. samples $\boldsymbol \sigma \sim \mathcal D$, there exists an 
        algorithm satisfying $\mathbb E_{\boldsymbol \sigma \sim \mathcal D}[\textsc{alg}]{} = O\left(\mathbb E_{\boldsymbol \sigma \sim \mathcal D}[\textsc{opt}] + \mathcal R \cdot \operatorname{Var}(\boldsymbol \sigma) \right)$.
    \end{enumerate}
\end{restatable}

\begin{restatable}{theorem}{MTSthm}{(Informal).} \label{thm:MTSthm} For MTS, predictions $\widehat{w}$ of an optimal dual solution satisfy
\begin{enumerate}[label=(\alph*)]
        \item (Usefulness) There exists an algorithm augmented with predictions $\widehat{w}$ which satisfies $\ALG \leq \OPT + \eta$.
        \item  (Stability) Even if $\widehat{w}$ is not the optimal dual for the instance at hand, but rather for a nearby instance, the error in the prediction is small.
        \item (Learnability) The predictions are learnable with polynomial sample complexity.
    \end{enumerate}
\end{restatable}

We prove usefulness and stability in Sections \ref{sec:laminarsetcover} and \ref{sec:MTS}, and we defer the learnability results to the appendix. Additionally, we show in Appendix \ref{sec:uselessduals} that unlike for the laminar version, dual predictions are not helpful for the standard LP relaxation of general set cover -- there are worst-case instances that remain hard even when an optimal dual solution is known upfront.

\paragraph{Robustness.} Our algorithms can be made robust against high prediction error by using standard algorithm combination techniques, which can be found in the appendix. 
\paragraph{Experiments.}
Complementing our theoretical results, Section~\ref{sec:experiments} presents a short experimental evaluation of our algorithms instantiated for the $k$-server problem and the parking permit problem. Using real-world data, we observe substantial performance benefits compared to traditional online algorithms.




\subsection{Related Work}
\paragraph{Related Prediction Setups.}
Dual variables and primal--dual reasoning are ubiquitous in online algorithms \cite{BuchbinderN09}. In learning-augmented algorithms, \citet{BamasMS20} presented a method to augment primal--dual algorithms using \emph{primal} predictions (i.e., action predictions). Conceptually similar to our dual predictions is the \emph{learned weights} method of \citet{LattanziLMV20} for scheduling with restricted assignment. \citet{LavastidaM0X21} show that these weights are stable (termed \emph{instance-robust} there) and learnable. They do not, however, correspond to duals of the primal minimization problem. \citet{LiX21} extend results from these works to scheduling on unrelated machines, using predictions consisting of both machine weights \emph{and} dual variables. The dual variables are used here only to reduce from the unrelated machines setting to related machines restricted assignment, and the resulting instance is solved using the machine weight predictions similarly to \citet{LattanziLMV20}. A recent line of work by Cohen and Panigrahi further generalizes these results to online allocation problems, and shows that the machine weights can be interpreted as dual variables of an \emph{auxiliary convex program} of maximizing entropy subject to feasibility and near-optimality constraints \cite{CohenP23,CohenP25}. Our focus is different: we design algorithms whose predictions consist \emph{purely} of dual variables of the underlying primal minimization LP itself (not of an auxiliary convex optimization problem). To the best of our knowledge, this viewpoint has not previously been developed as a general design principle for learning-augmented online minimization.

Dual predictions have appeared in other learning-augmented settings outside our scope. On the online \emph{maximization} side, learning dual ``shadow'' prices from data is a classic idea in stochastic and random-order models, notably in AdWords and related allocation problems~\cite{DevanurH09,FeldmanHKMS10,AnLMV24}. On the offline side, learned duals can be used to warm-start and accelerate primal--dual optimization, e.g.\ for matching \cite{DinitzILMV21,ChenSVZ22}. 

\paragraph{(Laminar) Set Cover.} The online set cover problem~\cite{AlonAABN09} gave rise to the powerful online primal-dual framework with its many applications to countless problems~\cite{BuchbinderN09}. Learning-augmented treatments of online covering include~\cite{BamasMS20,AnandGKP22,GuptaPSS22,GrigorescuLSSZ22,Zeynali0HW21,AzarPT23,AmeliSV25}. 

An important special case of \emph{laminar} set cover is the ski rental problem, studied extensively in the learning-augmented literature, e.g., \citet{PurohitSK18,GollapudiP19,WeiZ20,AngelopoulosDJKR24,dynamicpowermanagement,DiakonikolasKTVZ21}. The parking permit problem~\cite{Meyerson05} is the seminal leasing problem in online algorithms. It generalizes ski rental and can be modeled, up to a constant factor, as laminar set cover. In the learning-augmented setting, \citet{YaroslavA24} studied the special case of three permit types. A recent paper by \citet{AmeliSV25} highlights the parking permit problem as an application of their learning-augmented algorithm for covering problems. They achieve competitive ratios of $O(\log\eta)$ deterministically and $O(\log \log\eta)$ randomized, where $\eta$ is the $\ell_1$ prediction error. These guarantees are strongest in regimes of bounded time horizon; in long-duration instances, $\eta$ may grow proportionally with the time horizon, in which case the bound can be weaker than classical bounds that depend on the number of permits rather than the time horizon.

\paragraph{Metrical Task Systems.} In the classical setting without predictions, MTS is now well understood: For any $n$-point metric, the competitive ratio is exactly $2n-1$~\cite{BorodinLS92} for deterministic algorithms and $O(\log^2 n)$ for randomized algorithms~\cite{BubeckCLL21,CoesterL22,BubeckCR23}. Better bounds are achieved for special cases of MTS: For example, $k$-server is $O(k)$-competitive \cite{kServerConjecture}, convex body chasing in a $d$-dimensional space is $O(d)$-competitive \cite{ChasingConvexBodiesOptimally, ChasingConvexBodies-Anupam}, and width-$w$ layered graph traversal is $\Theta(w^2)$-competitive randomized \cite{randomizedLGT,BubeckCR23} and $\Theta(2^w)$-competitive deterministically~\cite{FiatFKRRV98,CoesterT26}. 

Learning-augmented MTS have been considered in the setting with action predictions: \citet{AntoniadisCEPS23online} give algorithms with competitive ratio parametrized by the prediction error, \citet{ChristiansonSW23} analyze tradeoffs between consistency (i.e., competitive ratio relative to the predicted solution) and robustness (i.e., competitive ratio relative to the offline optimum), \citet{AntoniadisCEPS23mixing} and \citet{CosaE25} consider scenarios with multiple action predictors, and \citet{SadekE24} study algorithms that use action predictions parsimoniously. Special cases of MTS that were studied in the learning-augmented setting include paging/caching \cite{LykourisV21,Rohatgi20,Wei20,AntoniadisCEPS23online,JiangPS22,BansalCKPV22,AntoniadisBEFHLPS23}, $k$-server~\cite{LindermayrMS25,ChristiansonSW23}, convex function chasing~\cite{ChristiansonHW22,LechowiczC0BHWS24}, and dynamic power management~\cite{dynamicpowermanagement}.


\section{Laminar Set Cover} \label{sec:laminarsetcover}


In the set cover problem, we have a ground set $\mathcal U = \{u_1, \dots, u_n\}$, and a set family $\mathcal S = \{S_1, \dots, S_m \}$ with $S_j \subseteq \mathcal U$ and associated costs $c_{S_j}\ge 0$. An algorithm must then find a cover $\mathcal C \subseteq \mathcal S$, satisfying $\bigcup _{S \in \mathcal C} S = \mathcal U$. The total cost an algorithm seeks to minimize is $c(\mathcal C) := \sum_{S \in \mathcal C} c_S$. In the \emph{online} set cover problem, only a subset of $X \subseteq \mathcal U$ must be covered, where the elements in $X$ are revealed sequentially in an online fashion. We use a discrete time model, where at time $t$, a request $e_t \in X$ is revealed and must be covered by some set in $\mathcal S$. Additionally, we suppose $\mathcal S$ is a laminar set family, meaning that any two $S, S' \in \mathcal S$ are either disjoint, or one set is a subset of the other.

The assumption of laminarity crucially means that the sets in $\mathcal S$ form a natural hierarchy, where sets higher in the hierarchy contain all sets lower in the hierarchy. We assume without loss of generality that $\mathcal U\in\mathcal S$. We represent $\mathcal S$ as a rooted tree $\mathcal T = (\mathcal S, E)$, where the vertices are the sets of $\mathcal S$, arcs $a = (S,S')$ represent that $S \supseteq S'$, and $\mathcal U$ is the root.

We further assume without loss of generality that if $S \subsetneq S'$, then $c_S < c_{S'}$, as otherwise in any feasible solution we can replace any $S$ with $S'$ without increasing the cost. 


An LP relaxation of online laminar set cover is given by

\begin{equation*}
\begin{aligned}[t]
\min \quad & \sum_{S \in \mathcal S} c_S x_S \\
\text{s.t.} \quad & \sum_{S: e_t \in S} x_S \ge 1, && \forall e_t \\
& x_S \ge 0, && \forall S \in \mathcal S.
\end{aligned}
\end{equation*}

Variable $x_S$ denotes the fraction of set $S$ bought. In the online setting, constraints are revealed one by one and an algorithm may only increase $x_S$ over time. We call this the \emph{fractional} version of the problem. 
It suffices to solve this version of the problem due to an online randomized rounding procedure of Lemma \ref{lemma:rounding-randomized-online}. Proofs omitted from this section, such as this one, are found in Appendix \ref{subsec:omittedproofs}.

\begin{restatable}{lemma}{lemonlinerandomizedrounding} \label{lemma:rounding-randomized-online}
    There exists a randomized online rounding algorithm, $\algint$, which given a deterministic online algorithm $\algfrac$ for the online fractional laminar set cover problem satisfies $\mathbb{E}[\algint] \leq 4\algfrac$.
\end{restatable}

The corresponding dual LP is given by
\begin{equation*}
\begin{aligned}[t]
\max \quad & \sum_{e \in \mathcal U}  y_{e} \\
\text{s.t.} \quad & \sum_{e \in S} y_{e} \le c_S, && \forall S \in \mathcal S \\
& y_{e} \ge 0, && \forall e \in \mathcal U
\\
& y_{e} = 0, && \forall e \not \in X.
\end{aligned}
\end{equation*}

We define the error of a prediction $\hat{\mathbf y}$ as $\eta := \inf_{\text{optimal dual }\mathbf y^*} \lVert \hat{\mathbf y}-\mathbf y^* \mathbf \rVert_1$. In the next subsections, we sketch the ideas to obtain the \emph{usefulness} and \emph{stability} parts of Theorem \ref{thm:laminarsetcover}. 

\subsection{Algorithm with Dual Predictions} \label{sec:algorithms}

In this section, we show that dual laminar predictions are useful, in the sense that they can improve the performance of online algorithms when they are sufficiently accurate. We present Algorithm \ref{alg:nonrobustLaminarDual} as a learning-augmented algorithm for laminar set cover using dual predictions. Algorithm \ref{alg:nonrobustLaminarDual} requires an $\mathcal R$-competitive classical online algorithm $A$ (without predictions) as part of its input. 
Our algorithm and its performance guarantee is also parameterized by a constant $0 < \alpha < 1$. We define the notion of $\alpha$-saturation to capture how close a dual constraint is to being tight.

\begin{definition}[$\alpha$-saturated]
    For any $\alpha \in (0,1)$, we say $S$ is $\alpha$-saturated by a solution $\mathbf y$ if $
        \sum_{e \in S} y_{e} \geq \alpha \cdot c_S$.
\end{definition}

\begin{algorithm}[h]
    \caption{Algorithm with dual predictions for laminar set cover}\label{alg:nonrobustLaminarDual}
    \begin{algorithmic}[1]

    \STATE Receive prediction $ \hat {\mathbf y} \in \mathbb R^\mathcal U$ for an optimal dual solution.
    
    \STATE $\mathbf x^1 \gets \mathbf 0 \in \mathbb R^{\mathcal S}$, $\mathbf x^2 \gets \mathbf 0 \in \mathbb R^{\mathcal S}$, $\mathbf x \gets \mathbf 0 \in \mathbb R^{\mathcal S}$
    \STATE Initialize classical online algorithm $A$, with competitive ratio $\mathcal R$.
    \FOR{each $e_t$ not yet covered by $\mathbf x$} \label{line:rainypermitlessdays}
    
        \IF{$\exists S \in \mathcal S$ with $e_t \in S$ being $\alpha$-saturated by $\hat{ \mathbf y}$ }
            \STATE Let $S$ be the inclusion-wise maximal such set
            \STATE $\mathbf x^1_S \gets 1$\label{line:type1} \COMMENT{Type 1 Purchase}
        \ELSE
        \STATE Pass $e_t$ into $A$.
        \STATE Let $\mathbf x^2$ be the output of $A$.\label{line:type2} \COMMENT{Type 2 Purchase}
        \ENDIF

    \STATE $\mathbf x \gets \mathbf x^1 + \mathbf x^2$ 
    \ENDFOR
    \end{algorithmic}
    \end{algorithm}

The following theorem implies the usefulness property of Theorem \ref{thm:laminarsetcover} by choosing $\alpha=1/(1+\epsilon)$.

\begin{restatable}{theorem}{lemdualbound} \label{lem:lemdualboundlaminar}
    The cost of Algorithm~\ref{alg:nonrobustLaminarDual} is bounded by
    \begin{equation*}
        \mathbb E[\textsc{alg}] \leq \frac{1}{\alpha} \textsc{opt} + \frac{1}{\alpha} \lVert (\hat{\mathbf y}-\mathbf y^*) ^+\rVert_1 + \frac{\mathcal R}{1-\alpha}\lVert (\mathbf y^*-\hat{\mathbf y}) ^+\rVert_1
    \end{equation*}
    for any optimal dual solution $\mathbf y^*$.
\end{restatable}

As labeled in Algorithm \ref{alg:nonrobustLaminarDual} on Lines \ref{line:type1} and \ref{line:type2}, we consider two types of costs: Type 1 costs and Type 2 costs. Type 1 costs are for purchasing permits which are $\alpha$-saturated by $\hat{\mathbf y}$, whereas Type 2 costs equal the cost of $A$. 
Thus, we break the analysis into two parts, bounding above the Type 1 costs and Type 2 costs separately. In particular, Theorem \ref{lem:lemdualboundlaminar} follows from Lemma \ref{lem:type1boundlaminar} and Lemma  \ref{lem:type2boundlaminar}.

\begin{restatable}{lemma}{lemtypeonebound} \label{lem:type1boundlaminar}
        The Type 1 costs of Algorithm \ref{alg:nonrobustLaminarDual} are at most $\frac{1}{\alpha} \textsc{opt} + \frac{1}{\alpha} \lVert (\hat{\mathbf y}-\mathbf y^*) ^+\rVert_1$.
\end{restatable}

\begin{proof}[Sketch of Proof]
    The high level idea is that the costs from Type 1 purchases either correspond to the optimal dual solution $\mathbf y^*$, or are due to some surplus $\lVert (\hat{\mathbf y}-\mathbf y^*) ^+\rVert_1$ which caused a set to be $\alpha$-saturated. We can thus charge these costs to those corresponding objects, losing a factor of $1/\alpha$ in the process. As $\lVert \mathbf y^* \rVert_1 \leq \textsc{opt}$, we obtain the result.
\end{proof}

\begin{restatable}{lemma}{lemtypetwobound} \label{lem:type2boundlaminar}
    The Type 2 costs of Algorithm \ref{alg:nonrobustLaminarDual} are at most $\frac{\mathcal R}{1-\alpha}\lVert (\mathbf y^*-\hat{\mathbf y}) ^+\rVert_1$.
\end{restatable}
\begin{proof}[Sketch of Proof]
    By complementary slackness, if a requested element is not part of an $\alpha$-saturated constraint, this implies a large error in $\lVert (\mathbf y^*-\hat{\mathbf y}) ^+\rVert_1$. Therefore, whenever we need to use the Type 2 purchases, we can charge this cost to corresponding parts of $ (\mathbf y^*-\hat{\mathbf y}) ^+$, losing a factor of $\mathcal R / (1-\alpha)$ in the process.
\end{proof}




\subsection{Stability of Optimal Dual Solutions} \label{sec:stability}

To prove stability, we present Algorithm \ref{alg:optimaldualalgorithm} as a simple algorithm for calculating an optimal dual solution for an offline instance. 
It is a greedy algorithm which descends the laminar set hierarchy $\mathcal S$ in a depth-first manner, and increases all dual variables in the leaves in $X$ uniformly until doing so would cause a dual constraint to be violated. We say a set $S \in \mathcal S$ is \emph{saturated} if $\sum_{e : e \in S \cap X} y_e = c_S$. The algorithm concludes when all dual variables $y_e$ with $e \in X$ belong to a saturated set.

\begin{algorithm}[h]
    \caption{Optimal Dual Computation}\label{alg:optimaldualalgorithm}
    \begin{algorithmic}[1]

    \STATE Initialize $\mathbf y \gets \mathbf 0$.
    \STATE Let $X \subseteq \mathcal U$ be the set of elements which must be covered, revealed offline.
    \STATE Let $\mathcal L$ be a depth-first search ordering of leaves in $\mathcal T$.
    \FOR{$S \in \mathcal L$ with $S \cap X$ non-empty}
    \STATE Increase all $y_e$ with $e \in S \cap X$ uniformly until an ancestor $A$ of $S$ becomes saturated.
    \ENDFOR
    \end{algorithmic}
\end{algorithm}

\begin{restatable}{lemma}{lemoptimaldual}
    Algorithm \ref{alg:optimaldualalgorithm} computes an optimal solution to the dual program.
\end{restatable}

Let $\mathbf y^{\textsc{alg}}(X)$ denote the result of Algorithm \ref{alg:optimaldualalgorithm} on problem instance $X$. Our next result shows that $\mathbf y^{\textsc{alg}}(X)$ is stable under perturbations of $X$.

\begin{restatable}{theorem}{lemstabilityresult}
\label{lem:stabilityresult}
    Algorithm \ref{alg:optimaldualalgorithm} satisfies $\lVert \mathbf y^{\textsc{alg}}(X) - \mathbf y^{\textsc{alg}}(X') \rVert_1 \leq 2 \beta |X \Delta X'|$, where $X \Delta X' := (X - X') \cup (X' - X)$ denotes symmetric difference, and $\beta := \max_{e \in X \Delta X'} \min_{S: e \in S} c_S$.
\end{restatable}

\section{Metrical Task Systems}
\label{sec:MTS}
In Metrical Task Systems (MTS), an algorithm controls a server initially located at a given state (i.e.,\ point) $s_0$ in a metric space $\mathcal{M} = (M, d)$. At each time $t \in \{1, \dots, T\}$, a request arrives in the form of a cost function $c_t\colon M \rightarrow \reals_{\geq 0} \cup \{\infty\}$. In response, the algorithm moves the server to a new state $s_t \in M$, paying movement cost $d(s_{t-1}, s_t)$ and service cost $c_t(s_t)$. The goal is to minimize the total cost incurred while serving the requests.

We can formulate MTS as an integer program by introducing variables $f_t({a,b}) \in \{0, 1\}$ for $a, b \in M$, $t=1,\dots,T$, where $f_t({a,b}) = 1$ means that the server is moved from state $a$ to state $b$ at time $t$ (just before serving request $c_t$). By relaxing the constraints $f_t({a,b}) \in \{0, 1\}$ to $f_t({a,b}) \geq 0$, we obtain the following linear program: 
\begin{alignat}{3}
\min\quad 
& \sum_{t=1}^{T} \sum_{a,b \in M} f_t(a,b)\bigl(d(a,b) + c_t(b)\bigr) \tag{Obj} \label{eq:MTSobjective} \\[0.5em]
\text{s.t.}\quad
& \sum_{b \in M} f_1(a,b) = \mathbf{1}[a = s_0]
\quad &&\forall a \in M \tag{C1} \label{eq:MTSconstraint1} \\[0.6em]
& \sum_{b \in M}\! \bigl(f_{t+1}(a,b) \!-\! f_t(b,a)\bigr) \!=\! 0
\, &&\!\!\!\!\!\!\forall a \!\in\! M,
\, t \!\in\! [T\!-\!1] \tag{C2} \label{eq:MTSconstraint2} \\[0.6em]
& f_t(a,b) \geq 0
\quad \forall a, b \in M,
\, t \in [T] \tag{C3} \label{eq:MTSconstraint3}
\end{alignat}

Here, $f_t({a,b})$ represents the amount of server mass moved from $a$ to $b$ at time $t$, and $\sum_a f_t({a,b})$ represents the server mass located at $b$ at time $t$. The constraints (\ref{eq:MTSconstraint1}) ensure that all the server mass is initially located at $s_0$, while the constraints (\ref{eq:MTSconstraint2}) ensure the conservation of server mass: mass coming into state $a$ at time $t$ must equal the mass leaving state $a$ at time $t+1$. 

By introducing variables $w_0(a)$ and $w_t(a)$, corresponding to the constraints (\ref{eq:MTSconstraint1}) and (\ref{eq:MTSconstraint2}), respectively, we obtain the following dual program:
\begin{equation*}
\begin{alignedat}{3}
\max\quad 
& w_0(s_0) \\[0.5em]
\text{s.t.}\quad
& w_{t-1}(a) \leq w_t(b) + c_t(b) + d(a,b),&   \\ &\hspace{2.2cm}\forall a, b \in M
 \quad \forall t \in [T-1], \\[0.6em]
& w_{T-1}(a) \leq c_T(b) + d(a,b),
\quad \forall a, b \in M.
\end{alignedat}
\end{equation*}

Observe that the optimal solution of the dual program can be computed in decreasing order of $t$, by setting each variable to the highest possible value which does not violate any constraint. By defining $w_T(a) = 0$ for all $a \in M$, we obtain the following recurrence for the optimal solution:
\begin{equation*}
\begin{aligned}
    w_{t-1}(a) = \min_{b \in M} \{d(a,b) &+ c_t(b) + w_t(b) \}, 
    \\&  \forall t \in \{1, \dots, T\} , \forall a \in M.
    \end{aligned}
\end{equation*}

Notably, $w_t(a)$ precisely captures the minimum possible cost to serve requests at times $t+1, \dots, T$, if the server is at state $a$ after serving the request at time $t$. In this sense, $w_t$ can be viewed as a time-reversed analogue of a \emph{work function}, a central object in the study of online algorithms~\cite{ChrobakL91}. While a work function summarizes the minimum cost of serving the requests revealed so far and can therefore be computed online, the optimal dual variables in our formulation depend on future requests and hence cannot be computed online.


\subsection{Algorithm with Dual Predictions}\label{sec:MTSAlgo}
Our algorithm using predictions of the dual variables $w_t(s)$ is very simple, and resembles the well-known $A^*$ search algorithm. Algorithm \ref{alg:greedy-algo-MTS} attempts to minimize the actual cost paid in the next move, plus the estimated future cost. 

\begin{algorithm}[h]
    \caption{Algorithm with dual predictions for MTS}
    \label{alg:greedy-algo-MTS}
    \begin{algorithmic}[1]
    \FOR{$t = 1, \dots, T$}
    \STATE Receive cost function $c_t$ and prediction $\predfw_t$
    \STATE Choose $s_t \in \argmin_{s \in M} {d(s_{t-1}, s) + c_t(s) + \predfw_t(s)}$
    \ENDFOR
    \end{algorithmic}
\end{algorithm}

Interestingly, if instead of $\predfw_t$ we used the work function corresponding to time $t$ (which corresponds to dual variables of a \emph{different} LP formulation), we would recover precisely the Work Function Algorithm, which is the optimal deterministic online algorithm for MTS \cite{MTSsurvey}.


Before analyzing the performance of our algorithm, we first describe how we measure the prediction error. We assume that the prediction at the last step is always the zero function (i.e., $\predfw_T = 0$), as the algorithm knows that the online instance ends after the request $c_T$. 

Let $B^c$ be the Bellman operator with cost function $c$:
\[
(B^c(w)) (s) = \min_{s'\in M} \{w(s') + d(s,s') + c(s')\},
\]
and let $\| \cdot \|_{\mathrm{sp}}$ denote the span seminorm
\[
\| v \|_{\mathrm{sp}} := \max(v) - \min(v), 
\]
where 
$\max(v) = \max_i v_i$ and $\min(v) = \min_i v_i$.

For an instance with cost functions $c_1, \dots, c_T$, we define the error of the prediction $\widehat{w}$ by 
\begin{align*}
\eta_t &= {\| {B^{c_t}(\predfw_t) - \predfw_{t-1}} \|}_{\mathrm{sp}} \quad \text{and} \quad
\eta = \sum_{t=1}^T {\eta_t}.
\end{align*}


Intuitively, $\eta_t$ quantifies the discrepancy between $\predfw_{t-1}$, our estimate of the future cost prior to observing the request $c_t$, and $B^{c_t}(\predfw_t) $, the corresponding estimate after observing $c_t$. We use the more forgiving norm $\| \cdot \|_{\mathrm{sp}}$ rather than $\| \cdot \|_\infty$, as each $\predfw_t$ only needs to be defined up to an additive constant, which has no impact on the behavior of the algorithm. Note that the optimal dual solution $w$ satisfies $w_{t-1} = B^{c_t}(w_t)$ by definition, and therefore yields $\eta = 0$. 

\begin{theorem}
\label{th:MTSBoundForBellmanError}
Algorithm \ref{alg:greedy-algo-MTS} is $\left(1 + \frac{\eta}{\OPT}\right)$-competitive.
\end{theorem}

\begin{proof}
By choice of $s_t$,
\[
d(s_{t-1}, s_t) + c_t(s_t)  + \predfw_t(s_t) = (B^{c_t}(\predfw_t))(s_{t-1}). 
\]
Rearranging and summing over $t$, we obtain
\begin{equation}
\label{eq:MTS-bound-eq1}
\begin{aligned}
\ALG &= \sum_{t=1}^T d(s_{t-1}, s_t) + c_t(s_t) \\
&= \sum_{t=1}^T \left((B^{c_t}(\predfw_t))(s_{t-1}) - \predfw_t(s_t)\right) \\
&= \predfw_{0}(s_{0}) - \predfw_{T}(s_{T}) + \sum_{t=1}^T (B^{c_t}(\predfw_t) - \predfw_{t-1})(s_{t-1}) \\
&\leq \predfw_{0}(s_{0}) - \predfw_{T}(s_{T}) + \sum_{t=1}^T \max \left(B^{c_t}(\predfw_t) - \predfw_{t-1}\right).
\end{aligned}
\end{equation}

On the other hand, if the offline optimum solution is $s^*_0 = s_0, s^*_1, \dots, s^*_T$, we have 
\[
d(s^*_{t-1}, s^*_t) + c_t(s^*_t)  + \predfw_t(s^*_t) \geq (B^{c_t}(\predfw_t))(s^*_{t-1}),
\]
and by similar calculations we obtain
\begin{equation}
\label{eq:MTS-bound-eq2}
\begin{aligned}
\OPT 
&\geq \predfw_{0}(s^*_{0}) - \predfw_{T}(s^*_{T}) + \sum_{t=1}^T \min \left(B^{c_t}(\predfw_t) - \predfw_{t-1}\right).
\end{aligned}
\end{equation}

 Since  $s_0 = s^*_0$ and $w_{T}(s_{T}) = w_{T}(s^*_{T}) = 0$, combining \eqref{eq:MTS-bound-eq1} and \eqref{eq:MTS-bound-eq2} gives
 \begin{align*}
         \ALG &\leq \OPT + \sum_{t=1}^T \| {B^{c_t}(\predfw_t) - \predfw_{t-1}} \|_{\mathrm{sp}}
         = \OPT + \eta. \qedhere
 \end{align*}
\end{proof}

\begin{figure*}[t]
    \centering
    \begin{subfigure}[t]{0.40\textwidth}
        \centering
    \includegraphics[width=\linewidth]{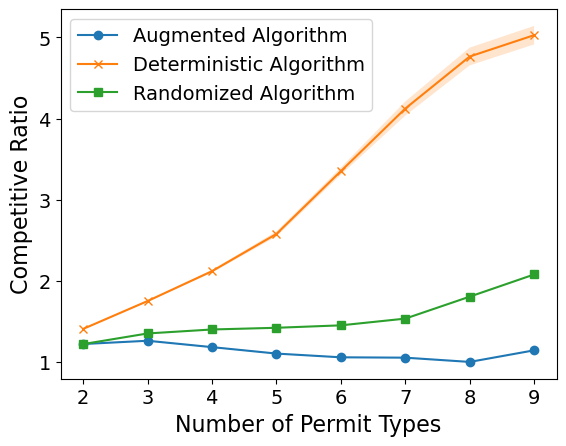}
        \caption{Experiment 1: Comparison of empirical competitive ratios as function of the number of permit types, $K$. Discount factor set to $f=1.5$.}
        \label{fig:varyingn}
    \end{subfigure}
    \hfill
    \begin{subfigure}[t]{0.40\textwidth}
        \centering
        \includegraphics[width=\linewidth]{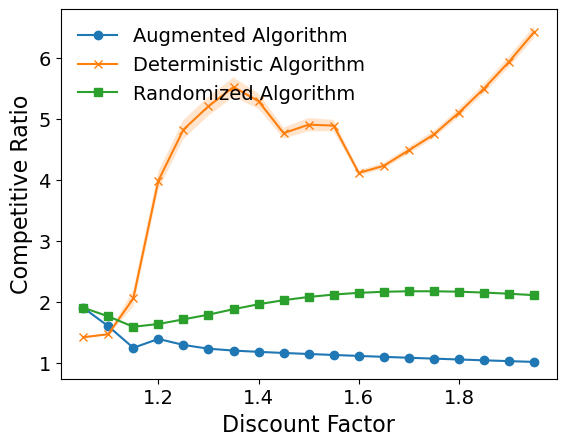}
        \caption{Experiment 2: Comparison of empirical competitive ratios obtained as a function of the Discount factor, $f$. Number of permits $K$ set to $K = 9$.}
        \label{fig:varyingdiscountfactor}
    \end{subfigure}
    \caption{Experiments on the Parking Permit Problem. Shades denote approximate 95\% confidence intervals.}
    \label{fig:pppexperiments}
\end{figure*}

\subsection{Stability of Dual Predictions}
The following theorem shows that predictions optimized for nearby instances have small error.
\begin{theorem}
\label{th:stabilityBellmanError}
Suppose that, on an instance with cost functions $c_1, \dots, c_T$, we use predictions $\predfw_1, \dots, \predfw_T$, which represent the optimal dual solution for an instance with cost functions $\widehat{c_1}, \dots, \widehat{c_T}$. Then,
$
\eta \leq \sum_{t=1}^T { \|\widehat{c}_t - c_t \|}_{\mathrm{sp}}.
$
\end{theorem}
\begin{proof}
    Fix $t\in \{1, \dots, T\}$. We know that $\predfw_{t-1} = B^{\widehat{c}_t}(\predfw_t)$, and therefore
    \begin{align*}
    \eta_t &= {\| {B^{c_t}(\predfw_t) - \predfw_{t-1}} \|}_{\mathrm{sp}} \\
    &= {\| {B^{c_t}(\predfw_t) - B^{\widehat{c}_t}(\predfw_t)} \|}_{\mathrm{sp}}.
    \end{align*}    
    We can now bound $\eta_t$ by ${ \|\widehat{c}_t - c_t \|}_{\mathrm{sp}}$. 
    Fix $s \in M$ and let $r \in M$ such that 
    \[
    (B^{\widehat{c}_t}(\predfw_t))(s) = \predfw_t(r) + d(s,r) + \widehat{c}_t(r).
    \]
    
    Then 
    \begin{equation}
    \label{eq:MTS-stability-eq-1}
    \begin{aligned}
    (B^{c_t}(\predfw_t))(s) &\leq \predfw_t(r) + d(s,r) + c_t(r)\\
    &= (B^{\widehat{c}_t}(\predfw_t))(s) + c_t(r) - \widehat{c}_t(r) \\
    & \leq (B^{\widehat{c}_t}(\predfw_t))(s) + \max(c_t - \widehat{c}_t).
    \end{aligned}
    \end{equation}
    One can similarly show that 
    \begin{equation}
    \label{eq:MTS-stability-eq-2}
    (B^{\widehat{c}_t}(\predfw_t))(s) \leq (B^{c_t}(\predfw_t))(s) - \min(c_t - \widehat{c}_t).
    \end{equation}
    Since \eqref{eq:MTS-stability-eq-1} and \eqref{eq:MTS-stability-eq-2} hold for all $s \in M$, we indeed have $\eta_t \leq { \|c_t - \widehat{c}_t \|}_{\mathrm{sp}}$. 
\end{proof}

\section{Experiments}\label{sec:experiments}

\subsection{Laminar Set Cover (Parking Permit Problem)}

We conduct an empirical evaluation of our laminar set cover algorithm by instantiating it on the Parking Permit Problem (PPP). Canonically in PPP, we consider an employee who must drive to work on rainy days, and otherwise they walk. When they drive to work, they must possess a valid parking permit to use the parking space. The dilemma is that the parking space has $K$ different types of permits, where each permit type $k  = 1,\dots, K$ lasts some fixed duration $D_k$ of consecutive days and costs some fixed amount $C_k$. The employee must decide whether it is more cost effective to purchase more expensive permits of longer durations, which may be cheaper on a per-rainy-day basis. The problem is a special case of online set-cover, where the days are elements of the universe, and the sets correspond to the days covered by each possible permit. \citet{Meyerson05} showed that the PPP can be reduced to Laminar Set Cover, while only losing a constant multiplicative factor in the competitive ratio.

While PPP is a general rent-or-buy problem, in our experiments we emulate the above scenario. We use 153 years of weather data measured between 1869 and 2021 at Central Park, New York City, acquired from the Global Historical Climatology Network daily (GHCNd) dataset \cite{ghcn}.

We create problem instances consisting of $T = 365$ days, motivated by the seasons of the year. We calculate the optimal dual vector for each instance, and use the point-wise sample mean over training instances as the predictor $\hat {\mathbf y}$. We evaluate the performance, as measured by the competitive ratio, using \emph{leave-one-out} cross validation, by leaving out each instance from the dataset and training on the others.

The permit types in our instances have $D_k = 2^k$ and $C_k = (2/f)^k$, where $f$ is a discount factor, representing the factor by which costs of subsequent permits decrease. We compare the performance of our learning-augmented algorithm (with $\alpha=0.5$) relative to state-of-the-art $O(K)$-competitive deterministic and $O(\log K)$-competitive randomized classical online algorithms for PPP \cite{coester2026a}. We use the randomized algorithm as the input $\mathcal R$-competitive algorithm in our learning-augmented procedure. 

\begin{figure*}[t]
    \centering
    \begin{subfigure}[t]{0.40\textwidth}
        \centering
    \includegraphics[width=\linewidth]{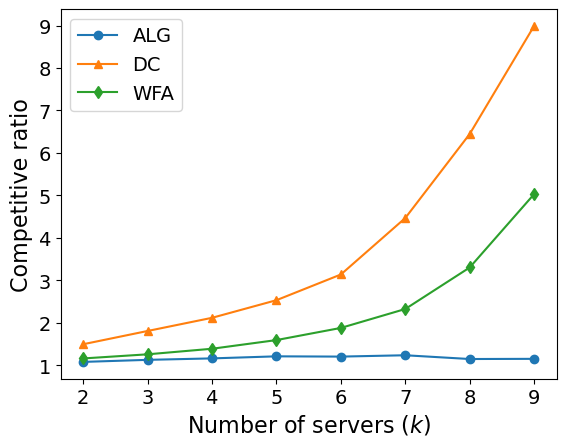}
        \caption{Experiment 1: A separate request for each trip.}
        \label{fig:plot1}
    \end{subfigure}
    \hfill
    \begin{subfigure}[t]{0.40\textwidth}
        \centering
        \includegraphics[width=\linewidth]{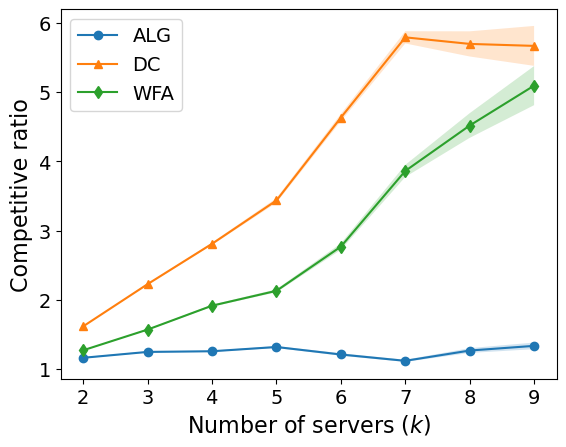}
        \caption{Experiment 2: One request per minute, at the point associated with the largest number of trips.}
        \label{fig:plot2}
    \end{subfigure}
    \caption{The average competitive ratio of each algorithm, where the number of points in the metric space is fixed to 10, and the number of servers varies from $k=2$ to $k=9$. Shades denote approximate 95\% confidence intervals.}
    \label{fig:two-plots}
\end{figure*}

The experimental results are displayed in Figures \ref{fig:varyingn} and \ref{fig:varyingdiscountfactor}. In Figure \ref{fig:varyingn}, we show how the empirical competitive ratio depends on $K$ for the different algorithms. We observe that the learning-augmented algorithm consistently has the best performance, and is near-optimal. The gap between the algorithms is more substantial as $K$ increases, where at $K=9$, the randomized algorithm has a competitive ratio $1.8$ times greater than the learning-augmented algorithm, and the deterministic algorithm is $4.4$ times greater. Figure \ref{fig:varyingdiscountfactor} explores the relationship between the discount factor and the algorithm's performance. For smaller discount factors of $f = 1.1$ or less, the deterministic algorithm performs better, likely because there is less penalty for conservatively choosing the cheaper permits. As $f$ grows, the gap between each algorithm grows, but this is not strictly monotonic. The conclusion from these experiments is that the learning-augmented algorithm is consistently useful on real-world data, particularly when the number of permits or the discount factors are large.

\subsection{Metrical Task Systems (k-Server)}

We evaluate the practicality of dual-based predictions for a particularly important special case of MTS, namely the k-server problem \cite{Manasse88}. In this problem, there are $k$ servers in a metric space, and requests arrive online at points in this metric space. Each request $r_t$ must be serviced by moving one of the servers to the requested location, and the goal is to minimize the total distance traveled by the servers. Note that $k$-server reduces to MTS on the metric space of server configurations, and where the cost function $c_t$ is $0$ at configurations containing $r_t$ and $\infty$ elsewhere. 

We conduct our experiments using a real-world dataset from a popular bike-sharing platform \cite{citibike_system_data}, which was previously used for benchmarks on a particular case of the $k$-server problem, namely caching~\cite{LykourisV21,AntoniadisCEPS23online}. We focus on trips originating in Manhattan and map each trip to a point in a metric space of size 10. Specifically, we partition Manhattan into 10 equally-sized intervals of latitude, which we associate to 10 points equally-spaced on a line, and then map each trip to the point corresponding to the latitude of its starting station. Using this mapping, we construct a $k$-server instance for each day. 

We perform two experiments on this dataset. In the first one, we issue a separate request for each individual trip. In the second experiment, we issue one request per minute of the day, corresponding to the point associated with the largest number of trips in that minute. In both experiments, we use data from 2023 and 2024 for training (731 instances in total) and data from 2025 for testing (365 instances). 

To learn predictions, we first compute the optimal dual variables for each training instance. We then partition time into blocks of 15 minutes and, for each instance and each block, compute the average of the optimal dual variables over that block. Next, we average these values across all training instances, yielding a single prediction per 15-minute block. During testing, we provide these learned average dual variables as predictions, updating them every 15 minutes. For example, all requests arriving between 1:00 and 1:15 are assigned the prediction obtained by averaging the dual variables corresponding to the 1:00–1:15 time block across all training instances. This choice of prediction reflects the idea that request patterns remain relatively stable within short time intervals. Moreover, aggregating predictions over 15-minute blocks yields sparser and more interpretable predictions, while also improving statistical robustness, since each learned prediction is based on multiple optimal dual solutions, both across days and within a single day. 

We compare the performance of our algorithm with two classical online 
algorithms: the Work Function Algorithm (WFA) and the Double Coverage (DC) algorithm. Both of them are known to achieve the optimal competitive ratio $k$ among deterministic online algorithms for the line metric space~\cite{ChrobakKPV90,BartalK04}.

The experimental results are shown in Figure \ref{fig:two-plots}. We observe that our algorithm consistently outperforms the classical online algorithms. Moreover, while our algorithm is nearly optimal for any number of servers $k$, the performance of WFA and DC degrades as function of $k$, in line with their theoretical $k$-competitiveness guarantees.

\section{Conclusion}
In this work, we demonstrated that for a rich assortment of problems, dual predictions possess several advantages over other typical prediction types.
We emphasize that the predictions considered by learning-augmented algorithms should not be so strong and difficult to learn that it merely outsources the difficult task to upstream learning algorithms; our stability and
learnability results show that dual predictions respect the practical
limitations of machine-learned predictors.

Future work could explore whether similar results hold for problems beyond MTS
and laminar set cover. For general set cover, worst-case instances remain hard
even when an optimal dual of the standard LP is known upfront
(cf.~Appendix~\ref{sec:uselessduals}), but a different LP formulation might
yield more useful duals -- just as our MTS results rely on a particular
formulation, while a very similar one would have been unsuitable
(Section~\ref{sec:MTSAlgo}).

\section*{Acknowledgements}

We thank the anonymous reviewers for their valuable comments.

We wish to give special thanks to Junhao Gan, William Umboh and Yixiang Wang, who found that Lemma \ref{lemma:rounding-randomized-online} isn't implicit in the work of \citet{Meyerson05}, as was claimed in our original manuscript. We also wish to thank Marcin Bieńkowski, Jarek Byrka and Yongho Shin, whose post-conference discussions led us to discover that WPAP can be reduced to laminar set cover.

Funded by the European Union (ERC, CCOO, 101165139). Views and opinions expressed are however those of the author(s) only and do not necessarily reflect those of the European Union or the European Research Council. Neither the European Union nor the granting authority can be held responsible for them.

\section*{Impact Statement}

This paper presents work whose goal is to advance the field of learning-augmented algorithms. There are many potential societal consequences of our work, none
which we feel must be specifically highlighted here.

\bibliography{references}
\bibliographystyle{icml2026}

\newpage
\appendix
\onecolumn
\section{Further Theoretical Results} \label{app:further}

In this section, we present some further results, which serves to further contextualize our research. In Appendix \ref{sec:instabprimalpred}, we show that primal predictions do not have the same stability guarantees as dual predictions. In Appendix \ref{sec:unstablecaching}, we show that a popular prediction model for caching is unstable. Finally, in Appendix \ref{sec:uselessduals}, we show that for general (non-laminar) set cover, even with perfect optimal dual predictions, a learning-augmented algorithm cannot do better than the classical online algorithm in some cases, at least for the natural LP formulation.

\subsection{Instability of Primal Predictions} \label{sec:instabprimalpred}

Firstly, we showcase that optimal primal predictions often do not enjoy the same stability as optimal dual predictions. We show that for optimal primal predictions, there is no analogous stability theorem to Theorem \ref{lem:stabilityresult}.

\begin{lemma}
    There exist problem instances $X$ and $X'$ for laminar set cover, such that for any primal optima $\mathbf x$ and $\mathbf x'$, $\lVert \mathbf x - \mathbf x' \rVert _1/\lvert X \Delta X' \rvert$ can be made arbitrarily large, even when $\beta := \max_{e \in \mathcal U} \min_{S: e \in S} c_S=1$.
\end{lemma}
\begin{proof}
We use the notation as defined in the main text. Firstly, let $N$ be an arbitrarily large integer. We consider an $(N+1)$-element universe $\mathcal U = \{0, 1, \dots, N \}$, and a collection of subsets $\mathcal S = \{\mathcal U, S_0, \dots, S_N \}$, where $S_i = \{i \}$ for $0 \leq i \leq N$. Then, define $c(\mathcal U) = N+\frac{1}{2}$, whereas $c(S_i) = 1$. Observe that for all $N$, $\beta$ is fixed at $1$. 

Now, consider $X = \mathcal U$ and $X' = \mathcal U - \{ 0\}$. Then $\lvert X \Delta X' \rvert  =1$. Observe that for $X$, the unique optimal solution is to buy $\mathcal U$, and on the other hand, the only optimal solution for $X'$ is $\{ S_1, \dots, S_N\}$. Therefore, the primal solution vectors $\mathbf x$ and $\mathbf x'$ differ in $N+1$ positions, meaning $\lVert \mathbf x- \mathbf x' \rVert_1 = N+1$. Thus, $\lVert \mathbf x - \mathbf x' \rVert _1/\lvert X \Delta X' \rvert = N+1$, and can therefore be made arbitrarily large.
\end{proof}

Next, we show that primal predictions are also unstable for $k$-server, which is a special case of MTS. In this problem, $k$ servers must be moved in a metric space $\mathcal{M} = (M, d)$ to serve incoming requests. At each time $t=1, \dots, T$, a request consisting of a point $r_t \in M$ arrives, and one of the servers must be moved to $r_t$. The cost incurred is equal to the total distance traveled by the servers.
We can reduce $k$-server to MTS on the metric space that contains all possible configurations, and the service costs are $0$ at configurations containing $r_t$ and $\infty$ elsewhere.

\begin{lemma}
    There exist problem instances $(r_t)_{t \geq 0}$ and $(r'_t)_{t \geq 0}$ for $k$-server 
    such that $(r_t)$ and $(r'_t)$ differ in a single request, and the difference $\|f-f'\|_1$ in the corresponding optimal primal solutions can be made arbitrarily large.
\end{lemma}
\begin{proof}
     Consider the $3$-server problem on the line with all servers starting at $0$, and the request sequence repeats $-2$, $-1$, $2$, $1$ many times and ends with a final request at $-2$. Then the unique optimal solution keeps two servers stationary at $-2$ and $-1$ and the third server goes back and forth between $1$ and $2$. Changing the final request to $2$ completely changes/mirrors the unique optimal solution, with two servers stationary at $1$ and $2$ and the third going back and forth between $-2$ and $-1$.
\end{proof}

\subsection{Instability of Event Predictions for Caching} \label{sec:unstablecaching}


Perhaps the most well-known line of work within learning-augmented algorithms is Lykouris and Vassilvitskii’s model for paging (special case of MTS) with next-request time prediction. 
In this model \cite{LykourisV21}, each request to a page $p$ comes with a prediction of the next time that page $p$ is requested again. The best learning-augmented algorithm for this setting by \citet{Wei20} works as follows: on a cache miss, evict the page with the furthest-in-future \emph{predicted} next-request time, but default to a standard online algorithm ignoring predictions if it has smaller overall cost.

We show that these predictions are also unstable in a very strong sense.

\begin{lemma} Replacing just a \emph{single} request in an otherwise perfectly predicted sequence can lead not only to unbounded prediction error, but to the performance of the learning-augmented algorithm by \citet{Wei20} completely deteriorating to a competitive ratio of $\Omega(\log k)$, offering no improvement over an online algorithm without predictions.
\end{lemma}
\begin{proof}
 Consider the following request sequence:
\begin{itemize}
\item Stage 1: Request $1,2,\dots,k,2$.
\item Stage 2: Repeatedly request $2,3,\dots,k+1$ many times.
\item Stage 3: An adversarial sequence on the pages $2,3,\dots,k+2$ such that the online algorithm without predictions attains its worst-case competitive ratio ($\Omega(\log k)$ if randomized, $\Omega(k)$ if deterministic).
\end{itemize}

Now imagine the \emph{anticipated} request sequence is the same except the last request of Stage 1 is $1$ instead of $2$, and next-request time predictions are sampled according to this anticipated sequence.

The first request, at page $1$, is accompanied by a prediction that page $1$ will be requested again at time $k+1$. Thus, after the first $k$ requests fill the initially empty cache, the evict-furthest-predicted policy would refuse to ever evict page $1$ as it anticipates a next request to $1$ earlier than any other request (although in the true request sequence page $1$ is never requested again). But this would lead to cost tending to infinity during Stage 2, whereas the online algorithm without predictions would serve Stage 2 for constant cost by keeping pages $2,3,\dots,k+1$ in cache. This triggers the defaulting behavior, whereby predictions are ignored and the overall algorithm offers no improvement over the online algorithm without predictions.
\end{proof}

\subsection{Difficulty of Using Dual Predictions for  General Set Cover} \label{sec:uselessduals}

In this section, we show that for the standard LP formulation of fractional set cover (i.e the same as the laminar set cover formulation in the main text), predictions of the optimal dual are not useful, in the sense of the following result.

\begin{lemma}
    A learning-augmented algorithm for online fractional set cover, with perfect predictions of the optimal dual, has competitive ratio of at least $H_m$, where $H_k$ is the $k$th Harmonic number and $m$ is the number of sets, $\lvert \mathcal S \rvert$.
\end{lemma}

\begin{proof}
We first present a proof of how a classical $H_m$ lower-bound on fractional set cover can be found (cf.~\citet{AlonAABN09}), and then show that perfect dual predictions are unhelpful at improving this.
Fix a deterministic online algorithm $\mathcal{A}$. We exhibit an adaptive adversarial instance with $m$ sets $S_1, \dots, S_m$ with $c(S_i) = 1$, on which $\mathcal{A}$ incurs cost at least $H_m$ while $\textsc{opt} = 1$. For our construction, we require that
for any non-empty subset of the sets $S_1, \dots, S_m$,
there exists an element which appears precisely in those subsets. 
\smallskip
\noindent\textit{Construction.} The adversary maintains a nested sequence of \emph{alive} index sets
\[
A_0 \supsetneq A_1 \supsetneq \cdots \supsetneq A_{m-1},
\]
with $A_0 = [m]$ and $|A_k| = m - k$. Let $x^{(k)} \in \mathbb{R}_{\geq 0}^m$ denote the algorithm's variable vector immediately after phase $k$. The vectors are nondecreasing: $x^{(k)} \geq x^{(k-1)}$ componentwise.
In phase $0$, the adversary presents element $e_0$ with
\[
\{i \in [m] : e_0 \in S_i\} = A_0 = [m],
\]
and the algorithm responds with $x^{(0)}$ satisfying $\sum_{i \in A_0} x^{(0)}_i \geq 1$.
For phases $k = 1, 2, \ldots, m-1$:
\begin{enumerate}
    \item The adversary selects
    \[
    i_k \in \arg\max_{i \in A_{k-1}} x^{(k-1)}_i,
    \]
    breaking ties arbitrarily, and sets $A_k := A_{k-1} \setminus \{i_k\}$.
    \item The adversary presents element $e_k$ with
    \[
    \{i \in [m] : e_k \in S_i\} = A_k.
    \]
    The algorithm responds, producing $x^{(k)} \geq x^{(k-1)}$ satisfying
    \[
    \sum_{i \in A_k} x^{(k)}_i \geq 1. \tag{$\ast_k$}
    \]
\end{enumerate}
\smallskip
\noindent\textit{Bounding $\textsc{opt}$.} Let $i_m$ denote the unique element of $A_{m-1}$. By construction, $i_m \in A_k$ for every $k \in \{0, 1, \ldots, m-1\}$, hence $e_k \in S_{i_m}$ for all $k \in \{0, 1, \ldots, m-1\}$. The offline assignment $x^*_{i_m} = 1$ and $x^*_i = 0$ for $i \neq i_m$ is therefore feasible and hence $\textsc{opt} = 1$.
\smallskip
\noindent\textit{Bounding $\textsc{alg}$.} We may assume without loss of generality that $\mathcal{A}$ is \emph{minimal} in the sense that at the end of each phase $k \in \{0, 1, \ldots, m-1\}$,
\[
\sum_{i \in A_k} x^{(k)}_i = 1,
\]
i.e., the algorithm raises variables only as much as needed to satisfy $(\ast_k)$; any algorithm that overshoots only incurs greater cost, so a lower bound under this assumption implies the same lower bound in general.
In particular, phase $0$ starts from $x = \mathbf{0}$ and ends with $\sum_{i \in [m]} x^{(0)}_i = 1$, contributing cost $1$.
We claim that for every $k \in \{1, \ldots, m-1\}$,
\[
\sum_{i \in [m]} \bigl(x^{(k)}_i - x^{(k-1)}_i\bigr) \;\geq\; \frac{1}{m - k + 1}. \tag{$\dagger_k$}
\]
By the minimality assumption applied to phase $k-1$,
\[
\sum_{i \in A_{k-1}} x^{(k-1)}_i = 1. \tag{1}
\]
Since $i_k$ maximizes $x^{(k-1)}_i$ over $i \in A_{k-1}$ and $|A_{k-1}| = m - k + 1$,
\[
x^{(k-1)}_{i_k} \;\geq\; \frac{1}{m - k + 1} \sum_{i \in A_{k-1}} x^{(k-1)}_i \;=\; \frac{1}{m - k + 1}. \tag{2}
\]
Since $A_k = A_{k-1} \setminus \{i_k\}$, combining (1) and (2),
\[
\sum_{i \in A_k} x^{(k-1)}_i \;=\; \sum_{i \in A_{k-1}} x^{(k-1)}_i \;-\; x^{(k-1)}_{i_k} \;\leq\; 1 - \frac{1}{m - k + 1}. \tag{3}
\]
Constraint $(\ast_k)$ requires $\sum_{i \in A_k} x^{(k)}_i \geq 1$, so
\[
\sum_{i \in A_k} \bigl(x^{(k)}_i - x^{(k-1)}_i\bigr) \;\geq\; 1 - \sum_{i \in A_k} x^{(k-1)}_i \;\geq\; \frac{1}{m - k + 1},
\]
and since $x^{(k)} \geq x^{(k-1)}$ componentwise, this yields $(\dagger_k)$.
\smallskip
\noindent\textit{Summing.} Since the unit costs give $\textsc{alg} = \sum_i x^{(m-1)}_i$ and the increments telescope,
\[
\textsc{alg} \;=\; \sum_{i \in [m]} x^{(0)}_i \;+\; \sum_{k=1}^{m-1} \sum_{i \in [m]} \bigl(x^{(k)}_i - x^{(k-1)}_i\bigr) \;\geq\; 1 \;+\; \sum_{k=1}^{m-1} \frac{1}{m - k + 1}.
\]
The sum, reindexed by $j = m - k + 1$, equals $\sum_{j=2}^{m} \tfrac{1}{j} = H_m - 1$, and we obtain
\[
\textsc{alg} \;\geq\; H_m.
\]
\smallskip
\noindent Combining with $\textsc{opt} \leq 1$,
\[
\frac{\textsc{alg}}{\textsc{opt}} \;\geq\; H_m.
\]
Now, consider the prediction $\hat{{\mathbf y}}$ which is zero, except for being $1$ at the index $e_0$. Observe that this dual solution is feasible, and it is optimal as its objective equals 1, which matches the primal $\textsc{opt}$ described. But, observe that this prediction cannot be used to improve the performance of an algorithm on online fractional set cover on the described adversarial example, because the prediction ${\mathbf y}_{e_0}$ does not reveal any information (since $e_0$ is covered by all sets $S_i$), and the prediction $\hat{{\mathbf y}}$ is constant at zero everywhere else. 
\end{proof}

\subsection{Equivalence of Laminar Set Cover and Interval Set Cover}

Laminar Set Cover is clearly a special case of Interval Set Cover (also known as Weighted Path Augmentation (WPAP;~\citet{WPAP1}), since a laminar set family can be arranged as intervals on a line. The following lemma shows that a reduction holds also in the opposite direction, meaning the problems are equivalent up to a constant factor.

\begin{lemma}\label{lem:laminarization}
    Online Interval Set Cover can be reduced to Online Laminar Set Cover, with only a constant-factor degradation in competitive ratio. 
\end{lemma}

\begin{proof}
Let $\mathcal{I}$ be the set of intervals for an interval set cover instance, where $\boldsymbol \sigma$ is the online sequence of points in $\mathbb R$. First, we will make some simplifying assumptions about $\mathcal I$, which can only deteriorate the competitive ratio by constant factors: 

\begin{enumerate}
    \item We can assume that the costs of intervals are all powers of 4, as we can round up the costs of intervals to their next power of 4 while only deteriorating the ratio by a constant factor. We will call $I \in \mathcal I$ a level $j$ interval if $c(I) = 4^j$. By scaling the costs up, we assume without loss of generality that the lowest level is 0.
    \item For any interval $I$ with $J \supsetneq I$, we must have $c(J) > c(I)$.
    \item For each level $j$, the level-$j$ intervals may only intersect at their endpoints. This is technical enough to justify its own lemma, so we prove it in Lemma \ref{lem:intersectatendpoints}.
    
\end{enumerate}

Then we shall prove there is a laminar instance $\mathcal{I}'$ with the same ground set and sequence $\boldsymbol \sigma$, such that:
\begin{enumerate}
    \item[(1)]  (\emph{Solution transfer}) any feasible solution to
    $\mathcal{I}'$ yields a feasible solution to $\mathcal{I}$ of the
    same cost;
    \item[(2)] (\emph{Cost bound}) $\OPT(\mathcal{I}') \le 2 \cdot \OPT(\mathcal{I})$.
\end{enumerate}
Consequently, an $\alpha$-competitive online algorithm for laminar set
cover yields an $O(\alpha)$-competitive online algorithm for $\mathcal{I}$.

\textbf{Construction.}
We build $\mathcal{I}'$ by processing levels $0, 1, \dots$ in increasing
order. Let $H_{<j}$ denote the (already laminar, by induction) family
of truncated intervals built from levels $< j$. For each raw level-$j$
interval $J$:

\begin{itemize}
    \item Since $H_{<j}$ is laminar, at any point $p$ of the line, the
    intervals of $H_{<j}$ containing $p$ form a chain; hence there is at
    most one inclusion-wise maximal such element. Let $I_1$ be this
    maximal element at $J$'s left endpoint, if it exists and genuinely
    straddles $J$'s boundary (i.e.\ it is neither nested in $J$ nor
    contains $J$), and define $I_2$ symmetrically at $J$'s right
    endpoint. 
    \item Set $J' := J \setminus (I_1 \cup I_2)$, and add $J'$ to the
    level-$j$ part of $\mathcal{I}'$ (dropping it if it becomes empty) with the same cost.
\end{itemize}

Observe in this construction, that $J'$ crosses no interval of $H_{<j}$. Any $K \in H_{<j}$ crossing $J$'s boundary lies, by definition of
maximality and nesting in a laminar family, inside some maximal
element $M$ covering that boundary point, and $M \in \{I_1, I_2\}$
by the uniqueness argument above. Since $K \subseteq M$ and
$J' = J \setminus (I_1 \cup I_2)$ excludes all of $M$, $K$ is disjoint
from $J'$. Together with our level-wise disjoint assumption on $\mathcal I$, which rules out crossings
\emph{within} level $j$, this shows
$H_{\le j} := H_{<j} \cup \{J' : J \text{ at level } j\}$ is laminar.
Induction over $j$ gives a fully laminar $\mathcal{I}'$. We now move to proving that $\mathcal I'$ satisfies (1) and (2), finishing the proof.

\textbf{Part (1): solution transfer.}
Since $J' \subseteq J$, any request covered in $\mathcal{I}'$ by buying
$J'$ is covered in $\mathcal{I}$ by buying $J$ instead, at identical
cost.

\textbf{Part (2): cost bound.}
Fix an optimal solution $O$ to $\mathcal{I}$. We now seek to find a solution $O'$ with bounded cost. Replace each $J \in O$
recursively: buy $J'$, and recursively replace $I_1, I_2$ (each
themselves possibly truncated, at some level $< j$) in the same way.
Since $J = J' \cup I_1 \cup I_2$, this recursive expansion covers
everything $J$ did, so the resulting solution $O'$ is feasible for
$\mathcal{I}'$.

Let $f(j)$ bound the cost of the full recursive expansion of one
level-$j$ interval. Since $I_1, I_2$ have level at most $j-1$, we
obtain the recursion
\[
f(j) \;\le\; 4^j + 2f(j-1), \qquad f(0) = 1.
\]
Unrolling,
\[
f(j) \;\le\; \sum_{i=0}^{j} 2^i \cdot 4^{j-i}
\;=\; 4^j \sum_{i=0}^{j}\left(\tfrac{1}{2}\right)^i
\;\le\; 2 \cdot 4^j.
\]
Summing over $J \in O$ gives
$\mathrm{cost}(O') \le 2 \cdot \mathrm{cost}(O) = 2\,\OPT(\mathcal{I})$.
Since $O'$ is feasible for $\mathcal{I}'$,
\[
\OPT(\mathcal{I}') \;\le\; \mathrm{cost}(O') \;\le\; 2\,\OPT(\mathcal{I}). \qedhere
\]
\end{proof}

\begin{lemma}\label{lem:intersectatendpoints}
    Online Interval Set Cover can be reduced to an instance where intervals of the same level intersect only at end-points, with only a constant factor degradation in competitive ratio.
\end{lemma}
\begin{proof}
    Starting from $\mathcal I$, we construct $\mathcal I^*$ level by level. Fix a level $j$, write $\mathcal I_j$
    for the set of level-$j$ intervals, and let $C=[A,B]$ be a connected component of
    $\bigcup_{I \in \mathcal I_j} I$ (each component is an interval, and the components
    are handled independently). We greedily select a chain of intervals covering $C$:
    set $\beta_0 := A$, let $K_1 \in \mathcal I_j$ be an interval with left endpoint $A$
    whose right endpoint $\beta_1$ is maximal, and, as long as $\beta_t < B$, let
    \[
        K_{t+1} \in \operatorname*{arg\,max}\bigl\{\, b \;:\; [a,b] \in \mathcal I_j,\ a \le \beta_t \,\bigr\},
    \]
    with right endpoint $\beta_{t+1}$. Whenever $\beta_t < B$, connectedness of $C$
    guarantees that some $[a,b] \in \mathcal I_j$ satisfies $a \le \beta_t < b$, so
    $\beta_{t+1} > \beta_t$. In particular no interval is selected twice, the chain
    $K_1, \dots, K_k$ is finite, and it covers $C$. In $\mathcal I^*$, we replace the
    level-$j$ intervals of $C$ by the \emph{increments}
    \[
        P_t := [\beta_{t-1}, \beta_t], \qquad t = 1, \dots, k
    \]
    each of cost $4^j$. The increments partition $C$, so they intersect only at their
    endpoints, and their union is again $C$, so $\mathcal I^*$ covers the same points
    as $\mathcal I$. Moreover, $P_t \subseteq K_t$, since $K_t$ has left endpoint at
    most $\beta_{t-1}$ and right endpoint $\beta_t$.

    It remains to bound the loss in the competitive ratio. First, given an online
    algorithm $A^*$ for $\mathcal I^*$, we obtain an online algorithm $A$ for
    $\mathcal I$ by buying $K_t$ whenever $A^*$ buys $P_t$; since $P_t \subseteq K_t$
    and $c(P_t) = c(K_t)$, algorithm $A$ covers every point that $A^*$ covers and pays
    exactly as much. Second, we claim that
    $\mathrm{OPT}(\mathcal I^*) \le 2\,\mathrm{OPT}(\mathcal I)$. It suffices to show
    that every $I = [a,b] \in \mathcal I_j$ contained in $C$ satisfies
    $I \subseteq P_t \cup P_{t+1}$ for some $t$, as then every interval of an optimal
    solution for $\mathcal I$ can be replaced by at most two increments of the same
    cost, covering a superset of its points. Let $t$ be minimal with $a \le \beta_t$,
    so that $a \ge \beta_{t-1}$. If $b \le \beta_t$, then $I \subseteq P_t$ and we are
    done. Otherwise $\beta_t < b \le B$, so the greedy step defining $K_{t+1}$ exists
    and $I$ is among its candidates, whence $\beta_{t+1} \ge b$ and
    $I \subseteq [\beta_{t-1}, \beta_{t+1}] = P_t \cup P_{t+1}$.
\end{proof}

\section{Omitted Laminar Set Cover Results}

\subsection{Proofs Omitted from Main Text} \label{subsec:omittedproofs}

We first give our online rounding algorithm for laminar set cover to turn a fractional algorithm into a randomized integral one at the loss of only a constant factor. This confirms, in particular, a conjecture of~\citet{WPAP1} about the existence of such a rounding for WPAP. We note that the classical PPP rounding algorithm of \cite{Meyerson05} does not work in the more general setting of laminar set cover, as it crucially requires that requests arrive in a particular order.

The underlying idea of our rounding is to decompose purchasing probabilities into conditional probabilities. Specifically, viewing $r_S:=\max\{0, 1-\sum_{A\supsetneq S}x_A\}$ as the probability that no superset of $S$ has been purchased, we can think of $\frac{x_S}{r_S}$ as the probability of buying $S$ conditioned on buying no superset of $S$. Our actual rounding algorithm then boosts these conditional probabilities by a factor 2, purchasing set $S$ with conditional probability $p_S:=\min\{1, \frac{2x_S}{r_S}\}$. Similar boosting factors have been used previously, for example, in online rounding algorithms for caching variants~\cite{AdamaszekCER19,BansalNT21,CoesterLNT22} and $k$-server~\cite{CoesterC26}. A difference to those applications is that their boosting serves to repair feasibility, whereas in our case, feasibility would hold even without boosting, and we need it to facilitate cost accounting instead: Note that $p_S$ can increase even while $x_S$ remains unchanged, as a result of $r_S$ decreasing. A decrease of $r_S$ means that supersets of $S$ are being fractionally purchased, and boosting accelerates these purchases, which in turn allows us to bound the probability of integral purchases of $S$ by $4x_S$.

\lemonlinerandomizedrounding*
\begin{proof}
    We prove that Algorithm \ref{alg:rounding-laminar} fulfills the conditions of the lemma. 

    Firstly, we show that Algorithm~\ref{alg:rounding-laminar} is feasible. To show that a given $e_t$ is covered, consider the inclusion-wise minimal set $S \in \mathcal S$ which covers $e_t$. By feasibility of $\algfrac$, we have that $\sum_{A \supseteq S} x_A(t) \geq 1$. Then, $r_{S}(t) \leq x_S(t)$ and hence $p_S(t) \geq 1 \geq U_S$. Hence some set covering $e_t$ must be included in $\mathcal C$.

    Now, fix $S$, and let the final value of $x_S(t)$ be $X_S := x_S(T)$. Let $E_S$ denote the event that $S \in \mathcal C$ by time $T$. To finish the proof, it suffices to show that $\mathbb P[E_S] \leq 4 X_S$, as

    \begin{equation}
    \mathbb{E}[\algint] = \sum_{S \in \mathcal S} c(S) \cdot \mathbb P(E_S) \leq4\sum_{S \in \mathcal S} c(S) X_S = 4 \algfrac.
    \end{equation}

\begin{algorithm}[h]
    \caption{Rounding Algorithm for Laminar Set Cover}
    \label{alg:rounding-laminar}
    \begin{algorithmic}[1]
    \STATE Compute $\mathbf x$ online using a deterministic fractional algorithm $\algfrac$
    \STATE Let $x_S(t)$ be the value of $x_S$ at time $t$
    \STATE Let $r_S(t):=\max\{0, 1-\sum_{A \supsetneq S} x_A(t)\}$ 
    \STATE Let $p_S(t):= \min\left\{ 1,\frac{2x_S(t)}{r_S(t)} \right\}$ (set to 1 if $r_S(t) = 0$)
    \STATE Draw $U_S$ uniformly from $[0,1]$ independently for each $S \in \mathcal S$
    \STATE Set $\mathcal C \gets \varnothing$
    \FOR{$t = 1, \dots, T$}
    \STATE Let $S_1 \supsetneq \dots \supsetneq S_l$ be the sets in $\mathcal S$ covering $e_t$
    \FOR{$i=1, \dots, l$}
    \IF{$e_t$ not yet covered by $\mathcal C$}\label{line:sample}
        \STATE Add $S_i$ to $\mathcal C$ if $U_{S_i} \leq p_{S_i}(t)$ 
    \ENDIF
    \ENDFOR
    \ENDFOR
    \RETURN $\mathcal C$
    \end{algorithmic}
\end{algorithm}

We now demonstrate that the conditional probability is bounded by $\mathbb{P} \left[ E_S | U_S = u \right] \leq \frac{4X_S^2}{u^2}$. Let $\tau$ be the first time that $e_{\tau} \in S$ and $u \leq p_S(\tau)$. For $E_S | U_S = u$, one condition we require is that every strict superset $A \supsetneq S$ satisfies $U_{A} > p_A(\tau)$. As each $U_A$ is drawn independently, this implies that

\begin{equation} \label{eq:boundbychain}
\mathbb{P}[E_S | U_S = u] \leq \prod_{A \supsetneq S} (1-p_A(\tau))
\end{equation}

Let $S_1 \supsetneq \dots \supsetneq S_p = S$ be the entire chain of supersets of $S$. Now, if $p_{S_i}(\tau) = 1$ for some $i<p$, then $\mathbb{P} \left[ E_S | U_S = u \right] = 0$, and our desired bound follows trivially. Otherwise, $p_{S_i}(\tau) = \frac{2 x_{S_i}(\tau)}{r_{S_i}(\tau)}$ for all $i<p$. Then for $i<p$,

\begin{equation}
1-p_{S_i}(\tau) = 1-\frac{2 x_{S_i}(\tau)}{r_{S_i}(\tau)} \leq \left(1-\frac{ x_{S_i}(\tau)}{r_{S_i}(\tau)}\right)^2 = \left(\frac{ r_{S_i}(\tau) - x_{S_i}(\tau)}{r_{S_i}(\tau)}\right)^2 = \left(\frac{ r_{S_{i+1}}(\tau)}{r_{S_i}(\tau)}\right)^2
\end{equation}

and with $r_{S_1}(\tau) = 1$, Inequality~\eqref{eq:boundbychain} telescopes to $\mathbb{P}[E_S | U_S = u] \leq r_S(\tau)^2$.

From here, observe that by definition of $\tau$, we must have $u \leq p_S(\tau) \leq \frac{2X_S}{r_S(\tau)}$, meaning $r_S(\tau) \leq \frac{2X_S}{u}$. Combining with earlier inequalities, this implies $\mathbb{P}[E_S | U_S = u] \leq \frac{4X_S^2}{u^2}$. Immediately then, we also have $\mathbb{P}[E_S | U_S = u] \leq \min \{1, \frac{4X_S^2}{u^2} \}$. We are now in a position to bound $\mathbb{P}[E_S]$. Our proof is finished by analyzing two simple cases, where in each $\mathbb{P}[E_{S}] \leq 4X_S$.

\textbf{Case 1. ($X_S \geq \frac{1}{2}$)}

\begin{equation}
    \mathbb{P}[E_{S}] \le 1 \leq 2X_S.
\end{equation}

\textbf{Case 2. ($X_S \leq \frac{1}{2}$)}

\begin{equation}
    \mathbb{P}[E_{S}] \le \int_0^1 \min \left\{1, \frac{4X_S^2}{u^2} \right\}du = \int_0^{2X_S} du + \int_{2X_S}^1 \frac{4X_S^2}{u^2} du= 4X_S - 4X^2_S \leq 4X_S.
\end{equation}

\end{proof}

We now move to proving our learning-augmented results.

\lemtypeonebound*
\begin{proof}
The main part of the proof is proving that

\begin{equation} \label{ineq:type1actualineqlaminar}
    \text{Type 1 costs} \leq \frac{1}{\alpha} \lVert \hat{\mathbf y} \rVert_1
\end{equation}

We show this inequality using a charging strategy. Recall that for a Type 1 purchase for element $e \in X$, we include the set inclusion-wise maximal set $S$ which is $\alpha$-saturated by $\hat{\mathbf y}$, i.e $\sum_{f \in S} \hat{y}_{f} \geq \alpha \cdot c(S)$, equivalently $c(S) \leq \sum_{f \in S} \frac{\hat{y}_{f}}{\alpha}$. Hence, we can charge the Type 1 costs to $\frac{\hat{\mathbf y}}{\alpha}$. Observe that the laminar assumption makes it impossible for any double charges to occur. To see this, observe that if $S$ and $S'$ were two intervals which were both charged to, it is impossible that one is contained in the other, because Algorithm \ref{alg:nonrobustLaminarDual} purchases the inclusion-wise maximal $\alpha$-saturated set. However, if neither is contained in one another, then by the laminar assumption, they must be disjoint. Inequality \eqref{ineq:type1actualineqlaminar} therefore follows. 

Then, we can decompose $\lVert \hat{\mathbf{y}} \rVert_1$ using the simple inequality
\begin{equation} \label{ineq:unobscureOPTlaminar}
    \lVert \hat{\mathbf{y}} \rVert_1 \leq \lVert {\mathbf{y}} ^*\rVert_1 + \lVert (\hat{\mathbf{y}} -  {\mathbf{y}} ^*)^+\rVert_1.
\end{equation}
Finally, observe that $\lVert \mathbf y^* \rVert_1 = \sum_{e \in \mathcal U} y_e^* = \sum_{e \in X} y_e^*$, which is the dual objective value. By the strong-duality theorem, as $\mathbf y^*$ is an optimal solution to the dual, $\lVert \mathbf y^* \rVert_1$ equals the objective of the optimal solution to the primal program. Thus, $\lVert \mathbf y^* \rVert_1 = \textsc{opt}$. The string of inequalities therefore implies Lemma \ref{lem:type1boundlaminar}:
\begin{equation}
    \begin{aligned}
        \text{Type 1 costs} &\leq \frac{1}{\alpha} \lVert \hat{\mathbf y}\rVert_1 &\text{by Ineq. \eqref{ineq:type1actualineqlaminar}}\\ 
        &\leq \frac{1}{\alpha} \big(\lVert {\mathbf y}^*\rVert_1 + \lVert (\hat{\mathbf y}-{\mathbf y}^*)^+\rVert_1 \big) &\text{by Ineq. \eqref{ineq:unobscureOPTlaminar}}\\
        &= \frac{1}{\alpha} \textsc{opt} + \frac{1}{\alpha}\lVert (\hat{\mathbf y}-{\mathbf y}^*)^+\rVert_1.
    \end{aligned}
\end{equation}
\end{proof}

\lemtypetwobound*

\begin{proof}
    Firstly, fix an optimal dual solution $\mathbf y^*$. Then, let $\mathbf x^*$ be an optimal integral solution to the primal problem. Assume without loss of generality that each element $e \in X$ is covered by exactly one set for the cover $\mathbf x^*$. Let $X' \subseteq X$ be the collection of requested elements which do not belong to any $\alpha$-saturated set, and are therefore associated with Type 2 costs. Define the set $\mathbf x_2^* := \{S \in \mathbf x^* : S \cap X' \neq \varnothing\}$.

Now, by complementary slackness conditions, if $\mathbf x^*_S = 1$ then $\sum_{e \in S} y^*_e = c(S)$. Therefore, for all $S \in \mathbf x_2^*$ we have $\sum_{e \in S} y^*_s = c(S)$. On the other hand, for all $S \in \mathbf x^*_2$, the fact that $S \cap X' \neq \varnothing$ for some $e \in S$ implies that there was no $\alpha$-saturated set containing $e$, meaning $\sum_{e \in S} \hat{y}_e < \alpha \cdot c(S)$. Thus, for all $S \in \mathbf x^*_2$, we have $\sum_{e \in S} y^*_e - \hat{y}_e \geq (1-\alpha) c(S)$. This yields
\begin{equation}
    \begin{aligned}
        \sum_{S \in \mathbf x_2^*} c(S) &\leq \frac{1}{1-\alpha} \sum_{S \in \mathbf x_2^*}\sum_{e \in S} y^*_e - \hat{y}_e\\
        &\leq \frac{1}{1-\alpha} \sum_{S \in \mathbf x_2^*}\sum_{e \in S} (y^*_e - \hat{y}_e)^+\\
        &\leq \frac{1}{1-\alpha} \sum_{e \in \mathcal U} (y^*_e - \hat{y}_e)^+\\
        &= \frac{1}{1-\alpha} \lVert (\mathbf y^* - \hat{\mathbf y})^+ \rVert_1,
    \end{aligned}
\end{equation}
where the third inequality holds as we assumed that for each $e$, $\mathbf x^*$ covers $e$ by at most one set. Now, recall that $A$ is an online algorithm which receives some subset of $X'$ as an online sequence, and call the indicator function of this subset $\boldsymbol \sigma$. Define $\textsc{opt}_{\boldsymbol \sigma}$ as the cost of an offline optimal running on $\boldsymbol \sigma$. Observe that $\textsc{opt}_{\boldsymbol \sigma} \leq \sum_{S \in \mathbf x_2^*} c(S)$ as all $e_t $ in $\boldsymbol \sigma$ will be covered by $\mathbf x_2^*$, making it a feasible solution. As $A$ is run on the instance $\boldsymbol \sigma$ and is $\mathcal R$-competitive, we know $\mathbb{E}[\text{cost}(A)] \leq \mathcal R \cdot \textsc{opt}_{\boldsymbol \sigma}$. Chaining these inequalities together, we find
\begin{equation}
\begin{aligned}
    \text{Type 2 costs} &= \mathbb{E}[\text{cost}(A(\boldsymbol \sigma))]\\
        &\leq \mathcal R \cdot \textsc{opt}_{\boldsymbol \sigma}\\
        &\leq \mathcal R \cdot \sum_{S \in \mathbf x^*_2} c(S)\\
        &\leq \frac{\mathcal R}{1-\alpha} \lVert (\mathbf y^* - \hat{\mathbf y})^+ \rVert_1. \qedhere
\end{aligned}
\end{equation}
\end{proof}

\lemoptimaldual*
\begin{proof}
    
    For the purposes of this proof, let $\textsc{alg}$ denote Algorithm \ref{alg:optimaldualalgorithm}. We show for any $S \in \mathcal S$ that if \textsc{alg} is run on $\mathcal S |_S := \{S' \in \mathcal S : S' \subseteq S \}$, then \textsc{alg} calculates the optimal dual on this subproblem. Let $\textsc{alg}(S)$ denote the objective value of running $\textsc{alg}$ on the instance $\mathcal S|_S$. We proceed by induction on the height of $S$ in $\mathcal T$, denoted $h(S)$.

    For the base case when $h(S) = 0$, we are considering the leaves of $\mathcal T$. When \textsc{alg} is run on $\mathcal S |_S$, the algorithm will simply increase the variables $y_e$ with $e \in S \cap X$ until the objective value of $c(S)$ is reached. If $S \cap X = \varnothing$, then the algorithm will stop and the objective will simply be $0$. Either scenario is clearly optimal, concluding the base case.

    Consider some $S$, assume that the claim holds for heights less than $h(S)$. Now, if we run \textsc{alg} on $\mathcal S|_S$, then we consider two cases:
    
    \textbf{Case 1: }if during the running of the algorithm $S$ becomes saturated, then the objective value $c(S)$ is reached, which is optimal dual to the dual constraint that $\sum_{e \in S} y_e \leq c(S)$. 
    
    \textbf{Case 2:} suppose $S$ never becomes saturated during the running of \textsc{alg} on $\mathcal S|_S$. Then, the dual constraint for $S$ is never tight, meaning the algorithm is equivalent to running $\textsc{alg}$ on each of the children of $S$ independently. Let $S_1, \dots, S_c$ denote the children of $S$, hence $\textsc{alg}(S) = \sum_{i=1}^c \textsc{alg}(S_i)$. Let $\mathbf y$ be any feasible dual solution of $\mathcal S|_S$, then 
    \begin{equation}
        \sum_{e \in S \cap X} y_e = \sum_{i =1 }^c \sum_{e \in S_i \cap X} y_e \leq \sum_{i=1}^c \textsc{alg}(S_i) = \textsc{alg}(S)
    \end{equation}
    and therefore $\textsc{alg}(S)$ is also optimal for this case.
\end{proof}

Let $\mathbf y^{\textsc{alg}}(X)$ denote the result of Algorithm \ref{alg:optimaldualalgorithm} on problem instance $X$. Our next result shows that $\mathbf y^{\textsc{alg}}(X)$ is stable under perturbations of $X$, which is vital for the learnability of the optimal dual prediction.

\lemstabilityresult*
\begin{proof}
    We first consider the case that $X' = X \cup \{ e\}$ for some $e \in \mathcal U - X$, and seek to show that $\lVert \mathbf y^{\textsc{alg}}(X) - \mathbf y^{\textsc{alg}}(X') \rVert_1 \leq 2 c(S^e)$, where $S^e$ is the leaf node in $\mathcal T$ containing $e$. Let $\mathcal U = S_0, S_1, \dots, S_d = S^e$ be the path from the root to $S^e$ in $\mathcal T$, and let $\delta_i = \sum _{e \in S_i}[y^{\textsc{alg}}(X') - y^{\textsc{alg}}(X)]_e$ for $1 \leq i \leq d$ be the change in objective value between the two solutions. Observe that $\delta_i \geq 0$, as adding an extra element $e$ cannot decrease the objective value of each subproblem.

    We decompose $\lVert \mathbf y^{\textsc{alg}}(X) - \mathbf y^{\textsc{alg}}(X') \rVert_1 = \sum_{i=0}^d s_i$, where we charge cost associated with a change in the $f$th coordinate, $\lvert \mathbf y^{\textsc{alg}}(X)_f - \mathbf y^{\textsc{alg}}(X')_f \rvert$, to the deepest $S_i$ on the $\mathcal U$-$S^e$ path with $f \in S_i$. The crucial observation is that for $0 \leq i \leq d-1$, we have $s_i \leq \delta_{i+1} - \delta_{i}$, as this quantity represents the amount of mass that dual variables increased after $y_e$ lose due to the addition of $e$ in $X'$. Additionally, $\delta_d \leq s_d$ is clear - the increase in objective value at $\delta_d$ can only happen alongside a corresponding increasing in $s_d$, as all dual variables corresponding to $S^e$ are increased uniformly. Therefore, we have that $\lVert \mathbf y^{\textsc{alg}}(X) - \mathbf y^{\textsc{alg}}(X') \rVert_1 = \sum_{i=0}^d s_i \leq s_d +\delta_d - \delta_0 \leq 2 s_d$.

    The final step is to establish the bound $s_d \leq c(S^e)$. Let $n$ be the number of elements in $S^e \cap X$. If $n = 0$, then the only non-zero $\mathbf y^{\textsc{alg}}(X')_f$ for $f \in S_d$ occurs at $f=e$, as $X' \cap S^e = \{e \}$. Therefore, $s_d \leq c(S^e)$. On the other hand, for $n \geq 1$, the final value of $\mathbf y^{\textsc{alg}}(X')_e$ is at most $\frac{c(S^e)}{n+1}$, and at worst this amount of dual mass is taken from the elements of $X \cap S^e$. Therefore, $s_d \leq \frac{2 c(S^e)}{n+1} \leq c(S^e)$, thereby proving the claim.

    The fully general case specified by the Lemma follows by repeatedly applying this claim to all elements $e \in X \Delta X'$ and taking the maximum $c(S^e) = \min_{S:e \in S} c(S)$.
    \end{proof}

\subsection{Learnability} \label{subsec:laminarlearnability}

In this section, we show that the optimal dual prediction is efficiently PAC-learnable. We fix $\mathcal U$, $\mathcal S$ and $c : \mathcal S \rightarrow \mathbb R_{\geq 0}$. We let $\boldsymbol \sigma \in \{0,1 \}^\mathcal{U}$ encode $X \subseteq \mathcal U$, where $\sigma_e = 1$ indicates that $e \in X$. Now, fix a probability distribution $\mathcal D$ over the set of problem instances $\boldsymbol \sigma$. The performance of our learning-augmented algorithm, when using prediction $\hat{\mathbf y}$ when $\mathbf y^*$ is a true optimal, has cost bounded above in terms of $\lVert \hat{\mathbf y} - \mathbf y^* \rVert_1$. Therefore, we define the loss function
\begin{equation}
    L(\hat{\mathbf y}, \boldsymbol \sigma) = \lVert \hat{\mathbf y} - \mathbf y^*(\boldsymbol \sigma) \rVert_1
\end{equation}

and seek to find a prediction $\hat{\mathbf y}$ minimizing $\mathbb E_{\boldsymbol \sigma \sim \mathcal D}[L(\hat{\mathbf y}, \boldsymbol \sigma)]$. Let $\hat {\mathbf y}^* = \arg \min_{\mathbf y} [\mathbb{E}_{\boldsymbol \sigma \sim \mathcal D}L(\mathbf y, \boldsymbol \sigma)]$. 

\begin{restatable}{theorem}{thmlearnability} \label{thm:learnability}
    There is a learning algorithm that given $\epsilon, \delta > 0$ returns a dual prediction $\hat {\mathbf y}$ such that $\mathbb{E}_{\boldsymbol \sigma \sim \mathcal D}[L(\hat {\mathbf y}, \boldsymbol \sigma)] \leq \mathbb{E}_{\boldsymbol \sigma \sim \mathcal D}[L(\hat{\mathbf y}^*, \boldsymbol \sigma)]+ \epsilon$ with probability at least $1-\delta$. The algorithm has sampling complexity $O\left (\left(\frac{\beta}{\epsilon}\right)^2 (n \log n + \log(1/\delta)) \right)$ where $\beta := \max_{e \in \mathcal U} \min_{S: e \in S} c_S$, $n = |\mathcal U|$ and time complexity polynomial in $n, \frac{1}{\epsilon}, \frac{1}{\delta}$.
\end{restatable}

\begin{proof}
For every dual solution $\mathbf y$, we let $g_{\mathbf y}(\boldsymbol \sigma) = L(\mathbf y, \boldsymbol \sigma)$, and define $\mathcal H = \{g_{\mathbf y} : \mathbf y \in \mathbb R^T \}$. We use several concepts from statistical learning theory to prove Theorem \ref{thm:learnability}, namely that to prove learnability it is sufficient to upper-bound the \emph{pseudo-dimension} of $\mathcal H$.

\begin{definition}[Pseudo-dimension]
    Let $\mathcal H$ be a class of functions $h : X \rightarrow \mathbb R$ and $S = \{x_1, \dots, x_s \} \subseteq X$. Then $S$ is shattered by $\mathcal H$ if there exists real numbers $r_1, \dots, r_s$ so that for all $S' \subseteq S$ there is a function $h \in \mathcal H$ such that $f(x_i) \leq r_i$ iff $x_i \in S'$ for all $i$. The pseudo-dimension is the largest $s$ such that there exists an $S \subseteq X$ with $|S|=s$ which is shattered by $\mathcal H$.
\end{definition}

The relevance of pseudo-dimension comes from the following convergence result.
\begin{lemma}[\cite{sltref1, sltref2, sltref3}] \label{lem:statisticallearningtheoryres}
    Let $\mathcal D$ be a distribution over domain $X$ and $\mathcal H$ be a class of functions $h : X \rightarrow [0, H]$ with pseudo-dimension $d_\mathcal H$. For any $\epsilon, \delta > 0$, if we sample i.i.d $\mathbf x_1, \dots, \mathbf x_s \sim \mathcal D$ with $s \geq c_0 (\frac{H}{\epsilon})^2 (d_{\mathcal H} + \ln(1/\delta))$ it holds that

    \begin{equation}
        \Bigg \lvert \frac{1}{s} \sum_{i=1}^s h(\mathbf x_i) -\mathbb{E}_{\mathbf x \sim \mathcal D} [h(\mathbf x)]\Bigg \lvert \leq \epsilon
    \end{equation}

    for all $h \in \mathcal H$ with probability at least $1 - \delta$, for some universal constant $c_0$.
\end{lemma}

Using results from \cite{DinitzILMV21}, it is known that the pseudo-dimension of $\mathcal H$ is given by $d_\mathcal H = O(T \log T)$. With this, we are ready to prove Theorem \ref{thm:learnability}.

    Our learning algorithm works as follows. For fixed $\delta, \epsilon > 0$, we choose $s$ according to Lemma \ref{lem:statisticallearningtheoryres} so that the error rate is $\epsilon' = \frac{\epsilon}{2}$, and sample $\boldsymbol \sigma _1, \dots, \boldsymbol \sigma _s \sim \mathcal D$ independently. Observe that functions $h \in \mathcal H$ map within the range $[0, \beta]$, because non-negativity is a dual constraint, and no dual variable can exceed $\beta$ due to the dual constraints. As additionally $d_\mathcal H = O(n \log n)$, the sampling complexity follows.

    For each $\boldsymbol \sigma_i$, we calculate $\mathbf y^*( \boldsymbol \sigma_i)$ using the optimal dual algorithm. We then choose $\hat {\mathbf y} = \arg \min_{\mathbf y} \sum_{i=1}^s L(\mathbf y,  \boldsymbol \sigma_i)$. As we use the $L_1$ norm, this can be done by choosing the median value for each coordinate, which can clearly be done in time polynomial in $s$ and $T$. Finally, with probability at least $1-\delta$ we observe that

    \begin{equation}
        \begin{aligned}
            \mathbb E_{\boldsymbol \sigma \sim \mathcal D}[L(\hat{\mathbf y}, \boldsymbol \sigma)] &\leq \mathbb E_{\boldsymbol \sigma \sim \mathcal D}[L(\hat{\mathbf y}, \boldsymbol \sigma)] - \frac{1}{s}\sum_{i=1}^s L(\hat{\mathbf y}, \boldsymbol \sigma) +  \frac{1}{s}\sum_{i=1}^s L(\hat {\mathbf y}^*, \boldsymbol \sigma)\\
            &\leq \frac{\epsilon}{2} + \mathbb E_{\boldsymbol \sigma \sim \mathcal D}[L(\hat {\mathbf y}^*, \boldsymbol \sigma)] + \frac{\epsilon}{2}\\
            &= \mathbb E_{\boldsymbol \sigma \sim \mathcal D}[L(\hat {\mathbf y}^*, \boldsymbol \sigma)] + \epsilon
        \end{aligned}
    \end{equation}
\end{proof}

\subsection{Culminating Result} \label{sec:culmination}

In this section, we combine the competitiveness result, the stability result, and the learnability result into one ultimate theorem.

\begin{restatable}{theorem}{supertheorem} \label{thm:supertheorem}
    Let $\mathcal D$ be a distribution, where online sequences $\boldsymbol \sigma \sim \mathcal D$ are drawn independently, and let $\operatorname{Var}(\boldsymbol \sigma)$ denote the total variance of $\boldsymbol \sigma$, i.e., the trace of the covariance matrix. There is a polynomial-time learning algorithm that, given $O(\lvert \mathcal U\rvert \log \lvert\mathcal U \rvert)$ samples from $\mathcal D$, returns a dual laminar set cover prediction $\hat {\mathbf y}$ such that $\hat{\mathbf y}$ can be given to a learning-augmented algorithm, satisfying
    \begin{equation*}
        \mathbb E_{\boldsymbol \sigma \sim \mathcal D}[\textsc{alg}]{} = O\left(\mathbb E_{\boldsymbol \sigma \sim \mathcal D}[\textsc{opt}] + \mathcal R \cdot  \beta \operatorname{Var}(\boldsymbol \sigma){}  \right)
    \end{equation*}

    where $\beta := \max_{e \in \mathcal U} \min_{S: e \in S} c_S$ and $\mathcal R$ is the competitive ratio for online laminar set cover.
\end{restatable}

\begin{proof}
    One way to phrase Theorem \ref{lem:lemdualboundlaminar} is that we have the bound
    \begin{equation}
        \mathbb E[\textsc{alg}] = O\left(\textsc{opt}+ \mathcal R\cdot  L(\hat{\mathbf y}, \boldsymbol \sigma)\right)
    \end{equation}

    for any instance $\boldsymbol \sigma$, where the expectation is taken over the randomization of $\textsc{alg}$. By additionally randomizing over the $\boldsymbol \sigma \sim \mathcal D$, we obtain
    \begin{equation}
        \mathbb E_{\boldsymbol \sigma \sim \mathcal D}[\textsc{alg}] = O\left(\mathbb E_{\boldsymbol \sigma \sim \mathcal D}[\textsc{opt}]+\mathcal R \cdot{\mathbb{E}_{\boldsymbol \sigma \sim \mathcal D}[L(\hat{\mathbf y}, \boldsymbol \sigma)]}\right)
    \end{equation}

    By Theorem \ref{thm:learnability}, there is a learning algorithm which can achieve a dual laminar prediction $\hat{\mathbf y}$, which with probability $1-\delta$ can achieve $\mathbb{E}_{\boldsymbol \sigma \sim \mathcal D}[L(\hat {\mathbf y}, \boldsymbol \sigma)] \leq \mathbb{E}_{\boldsymbol \sigma \sim \mathcal D}[L(\hat{\mathbf y}^*, \boldsymbol \sigma)]+ \epsilon$ with the desired sampling complexity. Therefore, if we use $\hat {\mathbf y}$ as the dual prediction in Theorem \ref{lem:lemdualboundlaminar}, we obtain a bound 
    \begin{equation}
        \mathbb E_{\boldsymbol \sigma \sim \mathcal D}[\textsc{alg}] = O\left(\mathbb E_{\boldsymbol \sigma \sim \mathcal D}[\textsc{opt}]+\mathcal R \left( {\mathbb{E}_{\boldsymbol \sigma \sim \mathcal D}[L(\hat{\mathbf y}^*, \boldsymbol \sigma)]} + \epsilon\right) \right)
    \end{equation}
    
    We now seek to bound $\mathbb{E}_{\boldsymbol \sigma \sim \mathcal D}[L(\hat{\mathbf y}^*, \boldsymbol \sigma)]$ in terms of something more common. Firstly, observe that for fixed $x$ and random variable $Y$, by convexity of $\lVert \cdot \rVert_1$ Jensen's inequality yields

    \begin{equation} \label{ineq:jensens}
        \lVert x-\mathbb E_{Y}[Y] \rVert_1 = \lVert \mathbb E_{Y}[x-Y] \rVert_1 \leq \mathbb E_{Y}[\lVert x-Y \rVert_1] 
    \end{equation}
    
    Recall that $\hat{\mathbf y}^*$ is defined as the vector minimising $\mathbf y \mapsto \mathbb{E}_{\boldsymbol \sigma \sim \mathcal D}[L(\mathbf y, \boldsymbol \sigma)]$. Thus we get

    \begin{equation} \label{ineq:lossbound}
    \begin{aligned}
        \mathbb{E}_{\boldsymbol \sigma_1 \sim \mathcal D}[L(\hat{\mathbf y}^*, \boldsymbol \sigma_1)] &\leq \mathbb{E}_{\boldsymbol \sigma_1 \sim \mathcal D}[L(\mathbb{E}_{\boldsymbol \sigma_2 \sim \mathcal D}[\mathbf y^*(\boldsymbol \sigma_2)], \boldsymbol \sigma_1)]\\
        &= \mathbb{E}_{\boldsymbol \sigma_1 \sim \mathcal D}[\lVert \mathbf y^*(\boldsymbol \sigma_1) -\mathbb{E}_{\boldsymbol \sigma_2 \sim \mathcal D}[\mathbf y^*(\boldsymbol \sigma_2)] \rVert_1]\\
        &\leq \mathbb{E}_{\boldsymbol \sigma_1 \sim \mathcal D}\mathbb{E}_{\boldsymbol \sigma_2 \sim \mathcal D}[\lVert \mathbf y^*(\boldsymbol \sigma_1) -\mathbf y^*(\boldsymbol \sigma_2) \rVert_1] &\text{By Ineq. \eqref{ineq:jensens}} \\
        &\leq 2 \beta \cdot \mathbb{E}_{\boldsymbol \sigma_1 \sim \mathcal D}\mathbb{E}_{\boldsymbol \sigma_2 \sim \mathcal D}[\lVert \boldsymbol \sigma_1 -\boldsymbol \sigma_2 \rVert_1] &\text{By Lemma \ref{lem:stabilityresult}}
    \end{aligned}
    \end{equation}

    Let's now consider $\mathbb{E}_{\boldsymbol \sigma_1 \sim \mathcal D}\mathbb{E}_{\boldsymbol \sigma_2 \sim \mathcal D}[\lVert \boldsymbol \sigma_1 -\boldsymbol \sigma_2 \rVert_1] = \sum_{t=1}^T\mathbb{E}_{\boldsymbol \sigma_1 \sim \mathcal D}\mathbb{E}_{\boldsymbol \sigma_2 \sim \mathcal D}[\lvert  \sigma_{1t} - \sigma_{2t} \rvert]$, where $\sigma_{it}$ denotes the $t$th component of $\boldsymbol \sigma_i$. Observe that because $\sigma_{it} \in \{0,1 \}$, we have $\lvert  \sigma_{1t} - \sigma_{2t} \rvert$ equal to $1$ if and only if $\sigma_{1t} \neq \sigma_{2t}$. Let $p_t = \mathbb{P}_{\boldsymbol \sigma \sim \mathcal D}[\sigma_t = 1]$. Then, the probability that $\sigma_{1t} \neq \sigma_{2t}$ equals $2p_t(1-p_t)$, because either $\sigma_{1t} = 1$ and $\sigma_{2t}=0$, or $\sigma_{1t} = 0$ and $\sigma_{2t}=1$. Therefore $\mathbb{E}_{\boldsymbol \sigma_1 \sim \mathcal D}\mathbb{E}_{\boldsymbol \sigma_2 \sim \mathcal D}[\lvert  \sigma_{1t} - \sigma_{2t} \rvert] = 2p_t(1-p_t)$. Observe that because $\operatorname{Var}_{\boldsymbol \sigma \sim \mathcal D}[\sigma_t]= p_t(1-p_t) $, we have \[\mathbb{E}_{\boldsymbol \sigma_1 \sim \mathcal D}\mathbb{E}_{\boldsymbol \sigma_2 \sim \mathcal D}[\lVert \boldsymbol \sigma_1 -\boldsymbol \sigma_2 \rVert_1] = 2 \sum_{t=1}^T\operatorname{Var}_{\boldsymbol \sigma \sim \mathcal D}[\sigma_t] = 2 \operatorname{Var}_{\boldsymbol \sigma \sim \mathcal D}[\boldsymbol \sigma].\] Therefore, by \eqref{ineq:lossbound}, we ultimately find the bound of $\mathbb{E}_{\boldsymbol \sigma_1 \sim \mathcal D}[L(\hat{\mathbf y}^*, \boldsymbol \sigma_1)] \leq 4 \beta \operatorname{Var}[\boldsymbol \sigma]$, which therefore implies with probability $1 - \delta$ that

    \begin{equation}
        \mathbb E_{\boldsymbol \sigma \sim \mathcal D}[\textsc{alg}] = O\left(\mathbb E_{\boldsymbol \sigma \sim \mathcal D}[\textsc{opt}]+\mathcal R \left( \beta \operatorname{Var}[\boldsymbol \sigma] + \epsilon\right) \right)
    \end{equation}

    By choosing $\delta = 1 /  \mathcal R$, the remaining case with probability $\delta$ adds a negligible cost of $\leq \mathbb E_{\boldsymbol \sigma \sim \mathcal D}[\textsc{opt}]$ to the expected cost $\mathbb E_{\boldsymbol \sigma \sim \mathcal D}[\textsc{alg}]$. This also adds a cost of $\ln \mathcal R$ to the sampling complexity, but this is dominated by $ \lvert \mathcal U \rvert \ln \lvert \mathcal U \rvert$, as for any interesting case, we have $\lvert \mathcal U \rvert > \mathcal R$.
\end{proof}

\subsection{Robustness}

We remark that the above results provide no helpful guarantees when the error $\eta$ is large. An often desirable property is for learning-augmented algorithms to be \emph{robust} to large prediction errors, in the sense that the performance is no worse than classical online algorithms. For Laminar Set Cover, this can be achieved by the following Lemma.

\begin{restatable}{lemma}{lemcombination} \label{lem:detcombinationlemma}
    Suppose we have $m$ algorithms, $A_1, \dots, A_m$ for laminar set cover, each with competitive ratios $\alpha_1, \dots, \alpha_m$. Then there exists an algorithm with the $A_i$ as input and $\textsc{alg}$ denoting its objective value, which for any input sequence satisfies
    \begin{equation} \label{ineq:detcombbound}
        \textsc{alg} \leq m \cdot \min_{1 \leq i \leq m} \{\text{cost}(A_i(I)) \}
    \end{equation}
    and therefore $\textsc{alg}$ is $m \cdot \min_{1 \leq i \leq m} \{\alpha_i \}$-competitive. If each $A_i$ is deterministic, then so is $\textsc{alg}$.
\end{restatable}

\begin{proof}
    \begin{algorithm}[h]
    \caption{Deterministic laminar set cover combination algorithm}\label{alg:detcombinationalgorithm}
    \begin{algorithmic}[1]
        \STATE Initialise online algorithms $A_1, \dots, A_m$ for laminar set cover.

        \FOR{times $t = 1, \dots, T$}
            \STATE Let $\mathbf x_t^i$ be the online solution produced by algorithm $i$ at time $t$.
            \STATE Follow an algorithm $i$ whose $\mathbf x_t^i$ has minimum objective value.
        \ENDFOR
        
    \end{algorithmic}
    \end{algorithm}

    Algorithm \ref{alg:detcombinationalgorithm} fulfills the requirements of the lemma. For a particular instance $I$, without loss of generality, assume that $A_1$ is the final algorithm which $\textsc{alg}$ follows. Therefore, $\text{cost}(A_1(I)) = \min_{1 \leq i \leq m} \{\text{cost}(A_i(I))\}$. Let $i_t$ denote the index of the algorithm that $\textsc{alg}$ follows on day $t$, and let $\textsc{alg}_t$ denote the cost incurred by the algorithm on day $t$. The main inequality to observe is that, for any $1 \leq i \leq m$, we have
    \begin{equation}
        \sum_{t : i_t = i} \textsc{alg}_t \leq \text{cost}(A_1(I))
    \end{equation}

    This follows, because by definition $\textsc{alg}$ could only follow $A_i$ on days $t$ for which its accumulated cost is at most that of $A_1$ at that time, and because the costs are non-decreasing throughout time, that cost is at most $\text{cost}(A_1)$.
    Therefore, we have the chain of inequalities
    \begin{equation}
        \begin{aligned}
            \textsc{alg} &= \sum_{t=1}^T \textsc{alg}_t\\
            &= \sum_{i=1}^m \sum_{t : i_t = i} \textsc{alg}_t\\
            &\leq \sum_{i=1}^m  \text{cost}(A_1(I))\\
            &= m \cdot \min_{1 \leq i \leq m} \{\text{cost}(A_i(I)) \}. \qedhere
        \end{aligned}
    \end{equation}
\end{proof}

\section{Omitted Results for MTS}
\subsection{Learnability}
\label{sec:MTS-learnability}

In this section we establish that, under our error definition, our prediction class for MTS admits learnability with polynomial sample complexity. We fix a metric space $(M, d)$ and a time horizon $T$. We assume that the cost functions $c_1, \dots, c_T$ are drawn from a joint probability distribution $\mathcal{D}$, which is not known to us. We want to learn \[\predfw \in W := \{ (w_1, \dots, w_T) \mid w_t \in \reals^M \text{ and } w_t \text{ is 1-Lipschitz, for } t=1,\dots,T  \} \] such that, when our algorithm is provided with prediction $\predfw$, its expected cost on an instance drawn from $\mathcal{D}$ is minimized. By writing $\eta(\predfw; c)$ for the error of prediction $\predfw$ on instance $c$, we have $\ALG \leq \OPT + \eta(\predfw; c)$ by Theorem \ref{th:MTSBoundForBellmanError}, and therefore we seek to find $\predfw$ which minimizes $\E_{c \sim \mathcal{D}}[\eta(\predfw; c)]$.

\begin{theorem}
 Let $\mathcal{F} := \{ c \mapsto \eta(w; c) \mid w \in W\}$, and let $n = \vert{M} \vert$ and $D$ be the diameter of $M$. Then,
$\mathcal F$ is learnable with polynomial sample complexity: for a suitable value of $m$, polynomial in $n$, $T$, $D$, $1/\epsilon$, and $\ln(1/\delta)$, it holds that for every
$\epsilon>0$ and $\delta\in(0,1)$, if we draw $m$ i.i.d.\ samples $c^1,\dots,c^m\sim \mathcal D$ ,
then with probability at least $1-\delta$ every empirical risk minimizer
\[
\widehat w \in \arg\min_{w\in  W} \; R_S(w),
\qquad
R_S(w):=\frac1m\sum_{i=1}^m \eta(w;c^i),
\]
satisfies
\[
\mathbb E_{c\sim\mathcal D}[\eta(\widehat w;c)]
\;\le\;
\inf_{w\in W}\mathbb E_{c\sim\mathcal D}[\eta(w;c)]
\;+\;\epsilon.
\]
\end{theorem}

\begin{proof}
Let $S=(c^1,\dots,c^m)$ be i.i.d.\ samples from $\mathcal D$, and let $\widehat w$ be any empirical risk minimizer over $ W$. By Lemma~\ref{lem:pdim} we have $\mathrm{Pdim}(\mathcal F)=O(n^3T^2)$.
Moreover, for $w \in W$, $\eta(w;c)$ is bounded as follows:
\begin{align*}
    \eta(w; c) &= \sum_{t=1}^T {\| {B^{c}(w_t) - w_{t-1}} \|}_{\mathrm{sp}}.\\
    &\leq \sum_{t=1}^T {\| B^{c}(w_t) \|}_{\mathrm{sp}} + {\| w_{t-1}\|}_{\mathrm{sp}} \\ 
    &\leq 2DT.
\end{align*}
In the last step, we used the definition of the Bellman operator and the fact that $w_{t-1}$ is 1-Lipschitz.
Applying Lemma~\ref{lem:statisticallearningtheoryres}  yields that for
\[
m \;\ge\; c_0\left(\frac{2DT}{\epsilon}\right)^2\Bigl(\mathrm{Pdim}(\mathcal F)+\ln(1/\delta)\Bigr),
\]
with probability at least $1-\delta$,
\[
\sup_{w\in\mathcal W}\Bigl| \mathbb E_{c\sim\mathcal D}[\eta(w;c)] - R_S(w)\Bigr|
\;\le\; \epsilon/2.
\]
In the event that this bound holds, let $w^\star\in\arg\min_{w\in\mathcal W}\mathbb E[\eta(w;c)]$.
Using the optimality of $\widehat w$ and the uniform deviation bound twice, we get
\begin{align*}
\mathbb E[\eta(\widehat w;c)]
&\le R_S(\widehat w) + \epsilon/2
\le R_S(w^\star) + \epsilon/2
\le \mathbb E[\eta(w^\star;c)] + \epsilon.
\end{align*}
\end{proof}

\begin{lemma}
\label{lem:pdim}
     $\mathrm{Pdim}(\mathcal F)=O(n^3T^2 )$.
\end{lemma}
\begin{proof}
    We equivalently seek to bound the VC-dimension of the class $\mathcal{H} := \{ (c, y) \mapsto \mathbf{1}_{\eta(w; c) > y} \mid w \in W\}$. To this end, we rely on Theorem 2.3 in \cite{goldbergJerrum93}. As a prerequisite, we consider a membership testing algorithm which, on input $(w, c, y)$, outputs $1$ if $\eta(w; c) > y$ and $0$ otherwise. Recall that
    \begin{align*}
    \eta(w; c) &= \sum_{t=1}^T {\| {B^{c}(w_t) - w_{t-1}} \|}_{\mathrm{sp}}.
    \end{align*}
    The algorithm performs the computations in the straightforward way, and runs in time $\tau = O(T \cdot n^2)$. 
    Also, a concept in the class $\mathcal{H}$ can be represented by $T \cdot n$ real numbers. Therefore, by Theorem 2.3 in  \cite{goldbergJerrum93},
    \[
    \mathrm{Pdim}(\mathcal F)=\mathrm{VCdim}(\mathcal H)=O((Tn)\cdot(Tn^2))
=O(n^3T^2). \qedhere
    \]
\end{proof}

\subsection{Robustness}

We can make our learning-augmented algorithm robust by combining it with a classical online algorithm for MTS, by making use of the following theorem.

\begin{theorem}[\citet{AntoniadisCEPS23online}, generalization of \citet{kserverfiat}] 
\label{deterministicCombination}
Given algorithms $A_0, \dots, A_{m-1}$ for an MTS problem, there exists an algorithm which achieves cost at most 
\[
 O(m \cdot \min_{i=0}^{m-1} {\cost_{I}(A_i)})
 \]
 for any input sequence $I$. Moreover, if $A_0, \dots, A_{m-1}$ are deterministic, then the resulting algorithm is also deterministic.
\end{theorem}



\end{document}